\documentclass[aps,nofootinbib,english,prx,floatfix,amsmath,superscriptaddress,tightenlines,twocolumn]{revtex4-2}
\usepackage[pdftex,colorlinks=true,linkcolor=darkblue,citecolor=blue,urlcolor=darkred]{hyperref}
\usepackage[english]{babel}
\usepackage{braket}
\usepackage{amsmath,amsfonts, amssymb, amsthm, dsfont}
\usepackage{yfonts,bbm}
\usepackage{mathrsfs}
\usepackage{graphicx}
\usepackage[export]{adjustbox}
\usepackage[normalem]{ulem}
\usepackage{verbatim}
\usepackage{xcolor}
\usepackage{float}
\usepackage{enumitem}
\usepackage{mathtools}

\DeclareMathSymbol{:}{\mathord}{operators}{"3A}

\makeatletter
\def\maketitle{
\@author@finish
\title@column\titleblock@produce
\suppressfloats[t]}
\makeatother

\theoremstyle{plain}
\newtheorem{definition}{Definition}
\newtheorem{lemma}{Lemma}

\definecolor{darkblue}{rgb}{0.,0.,0.4}
\definecolor{darkred}{rgb}{0.5,0.,0.}

\usepackage{verbatim}
\usepackage{stackengine}

\usepackage{multirow}
\usepackage[capitalise]{cleveref}

\newcommand{\calC}{\mathcal{C}}
\newcommand{\calD}{\mathcal{D}}
\newcommand{\calE}{\mathcal{E}}
\newcommand{\calF}{\mathcal{F}}
\newcommand{\calG}{\mathcal{G}}
\newcommand{\calH}{\mathcal{H}}

\newcommand{\calM}{\mathcal{M}}

\newcommand{\calQ}{\mathcal{Q}}
\newcommand{\calR}{\mathcal{R}}
\newcommand{\calS}{\mathcal{S}}

\newcommand{\calU}{\mathcal{U}}

\renewcommand{\tilde}{\widetilde}

\newcommand{\bs}{\textbf{s}}

\newcommand{\tr}{\operatorname{tr}}

\newcommand{\norm}[1]{\left\lVert #1 \right\rVert}
\newcommand{\identity}{\mathbb{I}}

\newcommand{\polylog}{{\rm polylog}}

\newcommand{\myloop}{{\rm cl}}
\newcommand{\mytrivial}{{\rm tr}}
\newcommand{\trivialstate}{\rho_\mytrivial}
\newcommand{\conv}{\mathsf{conv}}
\newcommand{\sketbra}[1]{\ket{#1}\!\bra{#1}}
\newcommand{\ketbra}[2]{\ket{#1}\!\bra{#2}}
\newcommand{\eqL}{=}

\newcommand{\eqfig}[2]{\vcenter{\hbox{\includegraphics[height=#1]{#2}}}}

\definecolor{dkpurple}{rgb}{0.5,0.,0.5}

\begin{document}
\title{Mixed-state phases from local reversibility}

\author{Shengqi Sang}
\email{sangsq@stanford.edu}
\affiliation{Department of Physics, Stanford University, Stanford, CA 94305, USA}
\affiliation{Perimeter Institute for Theoretical Physics, Waterloo, Ontario N2L 2Y5, Canada}

\author{Leonardo A. Lessa}
\email{llessa@pitp.ca}
\affiliation{Perimeter Institute for Theoretical Physics, Waterloo, Ontario N2L 2Y5, Canada}
\affiliation{Department of Physics and Astronomy, University of Waterloo, Waterloo, Ontario N2L 3G1, Canada}

\author{Roger S. K. Mong}
\email{rmong@pitt.edu}
\affiliation{Department of Physics and Astronomy, University of Pittsburgh, Pittsburgh, Pennsylvania 15260, USA}

\author{Tarun Grover}
\email{tagrover@ucsd.edu}
\affiliation{Department of Physics, University of California at San Diego, La Jolla, California 92093, USA}

\author{Chong Wang}
\email{cwang4@pitp.ca}
\affiliation{Perimeter Institute for Theoretical Physics, Waterloo, Ontario N2L 2Y5, Canada}

\author{Timothy H. Hsieh}
 \email{thsieh@pitp.ca}
\affiliation{Perimeter Institute for Theoretical Physics, Waterloo, Ontario N2L 2Y5, Canada}

\begin{abstract}
We propose a refined definition of mixed-state phase equivalence based on locally reversible channel circuits.  We show that such circuits preserve topological degeneracy and the locality of all operators including both strong and weak symmetries. 
Under a locally reversible channel, weak unitary symmetries are locally dressed into channel symmetries, a new generalization of symmetry for open quantum systems.
For abelian higher-form symmetries, we show the refined definition preserves anomalies and spontaneous breaking 
of such symmetries within a phase.
As a primary example, a two-dimensional classical loop ensemble is trivial under the previously adopted definition of mixed-state phases. However, it has non-trivial topological degeneracy arising from a mutual anomaly between strong and weak 1-form symmetries, and our results show that it is not connected to a trivial state via locally reversible channel circuits.
\end{abstract}

\maketitle
\def\thefootnote{*$\dagger$}\footnotetext{S.S. and L.A.L. contributed equally to this work.}\def\thefootnote{\arabic{footnote}}

\tableofcontents

\section{Introduction}

Quantum information insights have led to substantial progress in understanding quantum many-body phenomena. Among them, a circuit-based definition of quantum phases of matter has become a cornerstone in classifying and characterizing gapped quantum phases of matter~\cite{chen2010local, hastings2005quasiadiabatic}. The definition, which states that two gapped ground states are in the same phase if they are related by a local unitary circuit, allows one to make rigorous statements about universal properties of quantum phases, especially those displaying topological order.

Recently, there has been progress in extending the notion of phase of matter from pure ground states to mixed states. This is motivated by the fact that physical systems can couple with an environment in increasingly controllable ways, enabling a large variety of non-trivial mixed states \cite{diss1, diss2, fan2023diagnostics, lee2023quantum, zou2023channeling, lu2023mixed, lee2022decodingmeasurementpreparedquantumphases, zhu2022nishimoris, balasubramanian2025localautomaton2dtoric}. An analogous circuit definition for mixed-state phases was proposed in \cite{hastings2011nonzero, coser2019classification}; namely, that two mixed states belong to the same phase if they can be related to each other through a pair of local channel circuits (Def.~\ref{def: two_way}, henceforth referred to as two-way channel definition).
Leveraging this definition, many fruitful results have been obtained for topological phases in mixed states~\cite{ma2023average, ma2023topological, ma2024symmetry, chen2023separability,chen2023symmetryenforced, xue_tensor_2024, lessa2024mixedstate, lessa2025higher, ellison2024towards, wang2023intrinsic, wang2024anomaly, wang_analog_2024, de2022symmetry, sang2023mixed, PhysRevX.14.041031, sang2024stability, zhou2025finite}.

A fundamental property of local unitary evolutions $U$ is that local operators $O$ remain local after the evolution, with $U^\dagger OU$ having slightly larger support compared to $O$. Preserving locality is essential for maintaining the defining properties of a phase, including topological order~\cite{hastings2005quasiadiabatic, bravyi2006lieb}. The locality-preserving 
nature of local unitary circuits $U$ is due to the fact that {\it individual} forward and backward gates outside of the light cone of $O$ cancel in local pairs; in this sense, $U$ is {\it locally reversible}. We take this additional property for granted in local unitary evolution.  However, for local channels, this additional requirement is strictly stronger: By definition, two-way channels are globally reversible, but not necessarily locally reversible, with respect to given input states. As a result, we find that two-way channels that are not locally reversible may not preserve the locality of all operations; in particular, we will demonstrate that certain weak symmetries of a mixed state may not remain local under two-way channels and lead to an unsatisfactory classification of phases.

Therefore, in this work, we propose a definition of mixed-state phases based on local reversibility, which is manifestly finer than the two-way channel definition. To guarantee local reversibility, we define two mixed states to belong to the same phase if one can be evolved into the other by a local channel circuit that preserves a finite {\it Markov length} throughout the evolution (See Def.~\ref{def: markov_path}). 
Informally, the Markov length of a mixed state captures how local is a region's correlation with the rest of the system. Its precise definition is via the conditional mutual information in the tripartition shown in Fig. \ref{fig: main}(b).
Crucially, the action of a local channel circuit in which the evolving state's Markov length remains finite is guaranteed to be locally reversible.


Using this refined definition of phase equivalence, we derive the following results, which may be of independent interest:   
\begin{enumerate}[leftmargin=.5cm]
\item We define in Sec. \ref{sec: TD} the topological degeneracy of mixed-states via local indistinguishability, which generalizes the topological degeneracy in topologically ordered ground states. We prove that topological degeneracy is invariant in each mixed-state phase defined using locally reversible channel circuits.
\item We prove in Sec. \ref{sec: anomaly} that the spontaneous symmetry breaking of 1-form symmetries, or equivalently the 't Hooft anomaly associated with the symmetries, is also a phase invariant. More concretely, if $\rho$ has spontaneously broken 1-form symmetries, then any other state $\sigma$ in the same phase as $\rho$ also possesses a dressed 1-form symmetry that is spontaneously broken. Remarkably, the dressed symmetry generically takes form of a quantum channel and we dub this a {\it channel symmetry}, whose explicit form is determined by the locally-reversible channel circuits connecting $\rho$ and $\sigma$.
\item We show in Sec. \ref{sec: TD_TE} how the broken weak 1-form symmetry can be detected by a non-zero topological entropy in the Levin-Wen partition~\cite{PhysRevLett.96.110405}, as well as by a non-trivial topological degeneracy on an annulus-shaped subregion. 
\end{enumerate}

We use the classical loop ensemble in 2D \cite{castelnovo2007topological, Zhang2025SWSSB} as our primary example for motivating and demonstrating the consequences of the refined definition. Despite lacking quantum coherence, the classical loop state exhibits topological properties, and, in particular, serves as a non-trivial classical memory. Indeed, we show that, while the loop state is classified as trivial in the two-way definition, it is non-trivial under our new definition via the above three results. Despite using this specific state for concreteness, our techniques and main findings are generally applicable.



\begin{figure}
    \includegraphics[width=\linewidth]{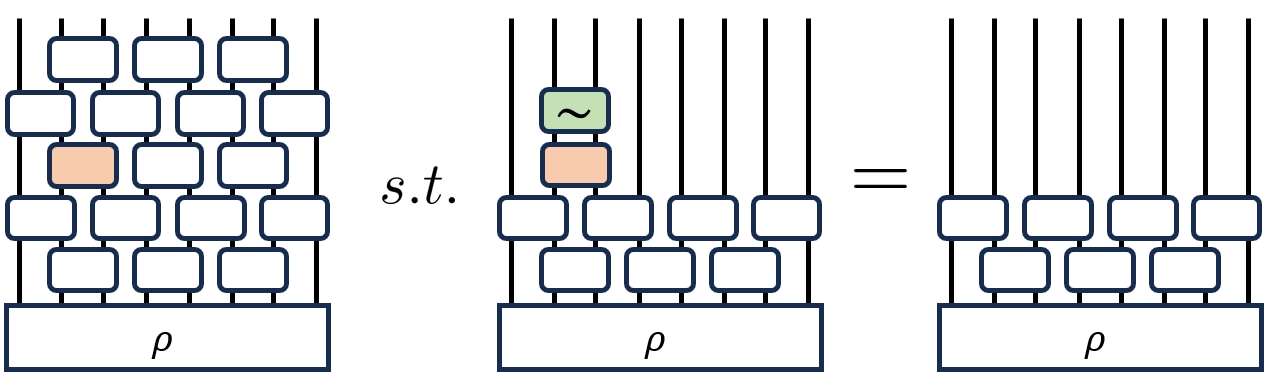}\\
    (a)\\
    \vspace{0.2cm}
    \includegraphics[width=.85\linewidth]{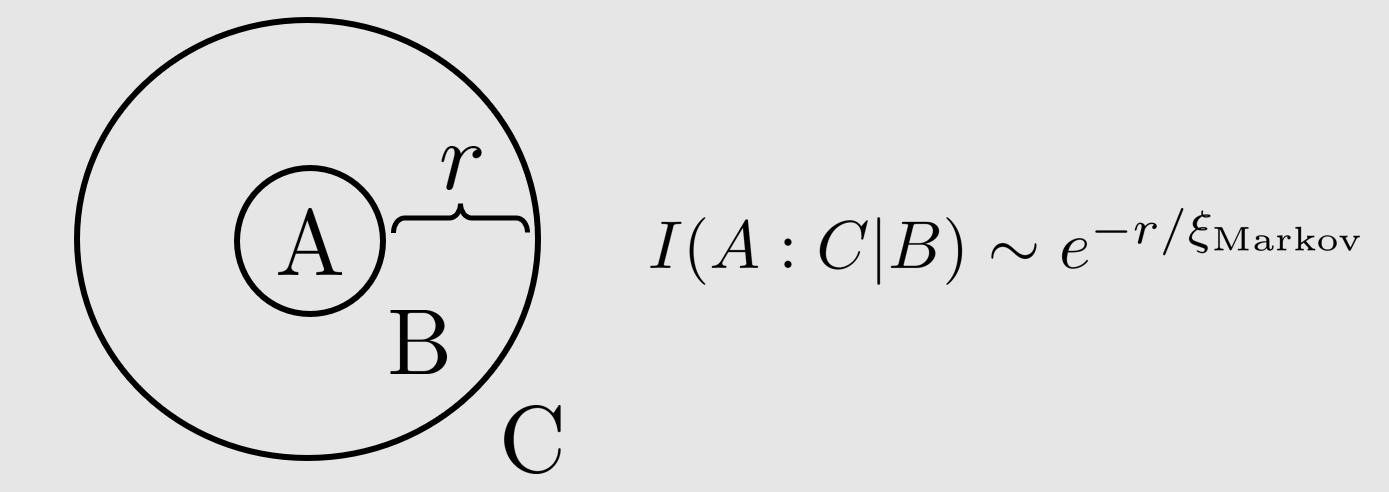}\\
    (b)
    \caption{
    (a) A channel circuit acting on $\rho$ is locally reversible if each of its gates  (\textit{e.g.} the one highlighted in red) is reversible, meaning that the gate's action upon the state at the time can be reversed locally by some other gate (the gate containing $\sim$). As such, local reversibility is a property of both the circuit and the input state $\rho$.
    Locally reversible circuits are natural extensions of local unitary circuits, which are locally reversible for any input state, to open many-body systems.
    (b) Markov length is defined as the length scale at which the conditional mutual information $I(A:C|B)$ decays with $B$'s width for the spatial partition shown in the figure. A sufficient condition for local reversibility is that the evolving state's Markov length $\xi_{\rm Markov}$ (Eq.\eqref{eq: markov_length}) is finite throughout the circuit. 
    } \label{fig: main}
\end{figure}

\section{Two definitions of mixed-state phase equivalence}\label{sec: defs}

\subsection{The two-way channel definition and the classical loop state}\label{sec: twoway}

Before introducing our refined definition of phase equivalence, we first review the one commonly used in most existing works, which we refer to as the two-way channel definition.

\begin{definition}[Two-way channel definition of phase equivalence~\cite{coser2019classification}]\label{def: two_way}
Two states $\rho_1$ and $\rho_2$ are in the same phase if there exist a pair of local channel circuits $\calE$ and $\calR$ such that:
\begin{equation}
\calE[\rho_1] \eqL \rho_2\quad\calR[\rho_2] \eqL \rho_1.
\end{equation}
\end{definition}
Here, a channel circuit is a circuit composed of quantum channels each supported on a small region, and ``local'' means that both the range of each individual channel in the circuit and the circuit depth are at most $\tilde O(1)$\footnote{$\tilde O(f(x))$ stands for $O(f(x)\polylog(x))$}.
We note that sometimes $\calE$ and $\calR$ are replaced by continuous time local Lindbladian evolutions (\textit{e.g.} in Ref.~\cite{coser2019classification}). In this work, we use the two notions interchangeably.

When we equate two expressions involving many-body states, the precise meaning is that their difference vanishes in the thermodynamic limit, namely:
\begin{equation}
x\eqL y\ \Leftrightarrow\ \lim_{L\rightarrow \infty}\norm{x-y} = 1/L^{-\infty}
\end{equation}
where $L$ is a length scale, usually taken as the linear size of the total system, or that of some subregion. $L^{-\infty}$ denotes a function that decays quicker than any $1/{\rm poly}(L)$. When $x$ and $y$ are numbers, $\norm{\cdot}$ denotes absolute value. When $x$ and $y$ are density matrices, $\norm{\cdot}$ denotes the operator trace norm. In the latter case, the super-polynomial decay of $\norm{x-y}$ guarantees that all local observables and information-theoretic quantities (\textit{e.g.}, von Neumann entropy) of $x$ and $y$ are also sufficiently close.

As mentioned in the introduction, the two-way channel definition has been successfully applied to extend various quantum orders to open systems, including topological order. On the other hand, since mixed states also encompass classical statistical mechanical ensembles, it is desirable for a complete definition to capture known classifications of classical statistical mechanics phases. However, the classification based on Def. \ref{def: two_way} is not consistent in some cases, as we will show.

Consider a 2D square lattice $\Lambda$ with one qubit defined on each edge $e$. A product state $\ket{\bs}$ in the Pauli $Z$ basis (\textit{e.g.}, $\ket{\bs}=\ket{01001...}$) is called a loop configuration if $A_v \ket{\bs} \equiv\prod_{e\in v}Z_e \ket{\bs}=\ket{\bs}$ is satisfied on each vertex $v$.
The \textbf{classical loop state} $\rho_{\myloop}$ on $\Lambda$ is defined as the uniform distribution of all possible loop configurations:
\begin{equation}
    \rho_{\myloop} \propto \sum_{\bs}\delta(\bs\text{ is a loop config.})\sketbra{\bs}
\end{equation}
If $\Lambda$ has a nonzero genus (\textit{i.e.}, $\Lambda$ has ``holes''), then we further require the summation to only include $\bs$ with even number of non-contractable loops around each hole.  For each hole, there is another state $\rho'_\myloop$ which is an ensemble over all configurations with odd number of non-contractible loops. 


When defined on a disk, the state $\rho_\myloop$ can be understood as the zero-temperature Gibbs state of the following Hamiltonian:
\begin{align}\label{eq: gibbs_rep_of_rholoop}
\begin{aligned}
    &H = {-\sum_{v} A_v\quad A_v\equiv \eqfig{.9cm}{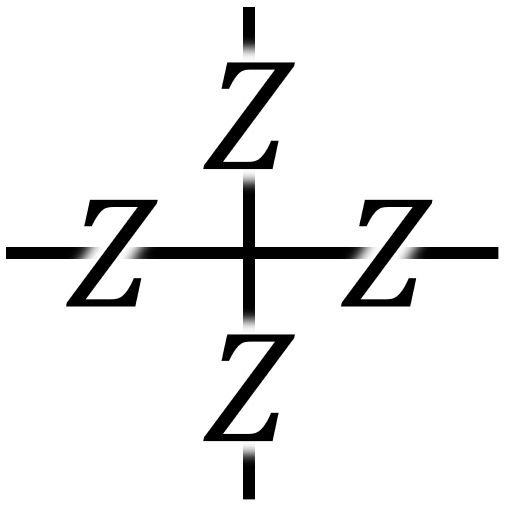}},\\
    &\rho_{\myloop}= \rho_{\beta=\infty} = \lim_{\beta\rightarrow\infty} e^{-\beta H} / \tr\bigl(e^{-\beta H}\bigr) .
\end{aligned}
\end{align}
This resembles a gauge theory, except that we are associating the kinetic energy of the gauge fields with vertices instead of plaquettes (on a square lattice this is equivalent to the usual formulation where kinetic energy is associated with plaquettes), and furthermore, we are not imposing the Gauss's law. 

The Gibbs state $\rho_{\myloop}$ possesses a strong 1-form symmetry, $\rho_\myloop = Z_{\overline{\omega}}\rho_\myloop$, and a weak 1-form symmetry, $\rho_\myloop= X_\gamma \rho_\myloop X_\gamma$, where $X_\gamma = {\prod}_{i\in\gamma} X_i$, $Z_{\overline{\omega}}={\prod}_{i\in\overline{\omega}} Z_i$, $\gamma$ and $\overline{\omega}$ represent any closed loop of edges on the lattice and dual-lattice, respectively. Crucially, the weak symmetry is spontaneously broken by $\rho_{\myloop}$ since the operators charged under it, such as $Z_{\overline{C}}$ where $\overline{C}$ is a non-contractible closed loop on the dual lattice, possess long-range correlations; we will elaborate on this point in much more detail in Sec.~\ref{sec: anomaly}. We thus expect $\rho_{\myloop}$ to be in a different phase from the infinite-temperature, maximally mixed ensemble $\rho_{\beta=0} \propto \identity$. 

Another hint pointing to the same conclusion comes from the topological entropy, \textit{i.e.} the subleading term in the von~Neumann entropy\footnote{In the previous literature, and especially in the context of pure states, the subleading term of the von Neumann entropy is usually called ``topological entanglement entropy''\cite{PhysRevLett.96.110404}. Here, we drop the word ``entanglement'' because the entropy of a mixed state can have classical contributions. In fact, $\rho_{\myloop}$ only has those, since it is fully unentangled.}. Comparing the topological entropy $\gamma$ of $\rho_{\myloop}$ with $\rho_{\beta=0}$, we find that the former has $\gamma = \log(2)$, while the latter has $\gamma=0$ \cite{castelnovo2007topological,Castelnovo2007entanglement,neg}. 

Likewise, $\rho_{\myloop}$ is also non-trivial from the perspective of error correction.  The two states $\rho_{\myloop}$, $\rho'_{\myloop}$ serve as a classical topological memory: they encode a classical bit of information in the value ($\pm 1$) of a non-contractible $Z_{\overline \omega}$ loop operator and they are locally indistinguishable, to be made precise shortly. 

However, the expectation that $\rho_{\myloop}$ is non-trivial is not fulfilled by the two-way definition of mixed-state phase equivalence (\textit{i.e.}, Def.\ref{def: two_way}). As we will now show, there exists a pair of $O(1)$-depth channel circuits that relate the product state $\rho_{\beta=0}$ to $\rho_{\myloop}=\rho_{\beta=\infty}$ and vice-versa. See also Appendix \ref{appendix:two-way_different_phases} for a thorough discussion on the phase transitions along both circuit paths.

To obtain $\rho_{\beta=0}$ from $\rho_{\beta=\infty}$, one can evolve the latter using a layer of bit-flip channel:
\begin{equation}
    \calE = \prod_e [\frac{1}{2}X_e(\cdot)X_e + \frac{1}{2}(\cdot)].
\end{equation}
The dynamics independently flips each qubit randomly and can be understood as the infinite-temperature Glauber dynamics of the model. It is straightforward to check that $\calE[\rho_{\myloop}]=\rho_{\beta=0}$.

The other direction ($\rho_{\beta=0}$ to $\rho_{\myloop}$)
is achieved in two steps. First, we turn $\rho_{\beta=0}$ into the product state $\ket{00 \cdots 0}$ by applying a layer of damping channel that turns any state into $\ket{00\cdots 0}$:
\begin{equation}
    \calR_1 = \prod_e [\tr(\cdot)\ketbra{0}{0}_e]
\end{equation}
Next, we apply the plaquette-flip channel
\begin{equation}\label{eq: lindbladian_pflip}
    \calR_2= \prod_p [\frac{1}{2}B_p(\cdot)B_p + \frac{1}{2}(\cdot)]
    \quad
    B_p \equiv \eqfig{.7cm}{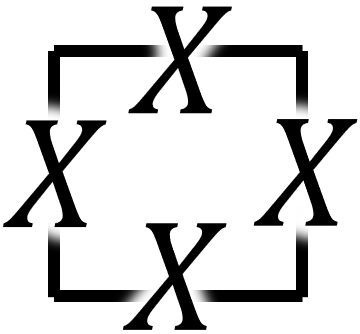},
\end{equation}
which randomly and independently flips edge qubits around each plaquette $p$. The dynamics turns the product state $\sketbra{\bf 0}$ into $\rho_{\myloop}$. We remark that both evolutions are completely classical, despite being represented as quantum channels.

The local channel $\calE$ does not preserve the topological memory (to be formally defined in Sec.\ref{sec: TD_of_rho_loop}) of the loop state, as it maps both $\rho_{\myloop}, \rho'_{\myloop}$ to the same state.  Thus, we see that global reversibility (two-way local channels) is not sufficient to preserve key properties of a phase. The underlying reason is that $\calE,\calR$ are not locally reversible when acting on their input states and do not preserve locality of certain symmetries of $\rho_{\myloop}$ (See Appendix \ref{appendix:local_non-reversibility}).


\subsection{Markov length and a refined definition}\label{sec: markov_length}

To address this shortcoming of Def.\ref{def: two_way}, we introduce a new definition of mixed-state phase that is ``finer'' than the previous one and is motivated by local reversibility.
As we will see in the rest of the work, this new definition characterizes $\rho_{\myloop}$ as being in a non-trivial phase and reveals several universal properties of mixed-state phases that the two-way definition Def.\ref{def: two_way} fails to capture.

We provide some motivation before introducing the definition. In the study of pure-state quantum phases, gapped ground states of local Hamiltonians naturally specify a set of `physical' many-body wavefunctions. The presence of an energy gap implies various universal properties shared by all gapped ground states (\textit{e.g.}, area-law scaling of entanglement and exponential decay of correlation), which play a crucial role in analyzing and classifying pure-state phases.

Physical mixed states, on the other hand, come in diverse types (Gibbs states, Lindbladian steady states, etc.) without an obvious analog of gapped ground states, and it is unclear whether one type includes other types. Therefore, it is more convenient to directly identify properties that should be generally satisfied by ``physical'' mixed states.


The \textbf{Markov length} provides one such property and is defined as follows. 
For a state $\rho$ and a tri-partition of its support $A|B|C$, the quantum conditional mutual information (CMI) $I(A:C|B)$ is defined as:
\begin{equation}
    I_\rho(A:C|B) = S_\rho(AB) + S_\rho(BC) - S_\rho(ABC) - S_\rho(B)
\end{equation}
where each $S_\rho(X)=-\tr(\rho_X\log(\rho_X))$ is the von Neumann entropy of region $X$.
We henceforth focus on the case where $\rho$ is defined on a 2D lattice that may have a non-trivial topology, and $A|B|C$ has a geometry as shown in Fig.\ref{fig: main}(b), where $A$ is a local simply-connected region, $B$ is an annulus-shaped buffer region surrounding $A$, and $C=\overline{B\cup A}$ is the rest of the system. We define $\rho$'s Markov length $\xi(\rho)$ as the smallest length $\xi$ such that:
\begin{equation}\label{eq: markov_length}
    I_\rho(A:C|B)\leq {\rm poly}(L) e^{-{\rm dist}(A, C)/\xi}
\end{equation}
is satisfied for all valid $(A,B,C)$ choices.
To gain intuition, we rewrite CMI as
\begin{equation}
    I_\rho(A:C|B)=I_\rho(A:BC)-I_\rho(A:B),
\end{equation}
where $I(A:B)\equiv S(A)+S(B)-S(AB)$ is the quantum mutual information, a measure of correlation between two parties. In this sense, the Markov length captures how local the correlation is between $A$ and the rest of the system $\overline A=BC$, because CMI being exponentially small in dist$(A,C)$ means that most of the correlation is captured by a $\xi(\rho)$-width buffer region around $A$.

Many important classes of pure and mixed states satisfy the finite Markov length property.
In the case of pure states $\rho=\sketbra{\psi}$, one has the relation:
\begin{equation}\label{eq: CMI_MI}
\begin{aligned}
I(A:C|B) &= \underbrace{S(AB)}_{S(C)}
+ \underbrace{S(BC)}_{S(A)} - \underbrace{S(B)}_{S(AC)} - \underbrace{S(ABC)}_{0\vphantom{)}} \\ &= I(A:C).
\end{aligned}
\end{equation}
For gapped ground states without spontaneous symmetry breaking, $I(A:C)$ usually decays exponentially with respect to ${\rm dist}(A:C)$, leading to a finite Markov length.
Furthermore, Gibbs states of commuting Hamiltonians (at any temperature) have exactly zero Markov length, because $I(A:C|B)$ is exactly zero after dist$(A,C)$ exceeds the interaction range of the Hamiltonian~\cite{brown2012quantum}. We note that the loop state $\rho_{\myloop}$ falls into the latter category (due to Eq.\eqref{eq: gibbs_rep_of_rholoop}) and has exactly zero Markov length. For the opposite behavior, states with strong-to-weak spontaneous symmetry breaking~\cite{Lessa2024, sala_spontaneous_2024, lee_quantum_2023} are notable examples of mixed states with diverging Markov length~\cite{Lessa2024}, even though they can have finite correlation length as defined from mutual information.

Based on the Markov length, we propose the following modified definition of mixed-state phase equivalence:
\begin{definition}[Markov length definition of phase equivalence]\label{def: markov_path}
Two states $\rho_1$ and $\rho_2$ are in the same phase if there exists a local channel circuit $\calC=\calC_T \cdots \calC_2\calC_1$ (each $\calC_t$ is a layer of non-overlapping local channel gates) such that:
\begin{align}
    \calC[\rho_1] \eqL \rho_2,
\end{align}
and, further, for any time $t$ and any $\calC'_t \subseteq \calC_t$ being a layer composed of a subset of gates in $\calC_t$,
$\calC'_{t}\calC_{t-1}\cdots\calC_1[\rho_1]$
has finite Markov length, \textit{i.e.}, $\xi(\rho_t)\leq\xi_{0}$ some $\xi_0=O(1)$ independent of $t$.
\end{definition}


It was shown in Ref.~\cite{sang2024stability} that two states being equivalent under Def.~\ref{def: markov_path} implies their equivalence under Def.~\ref{def: two_way}. More precisely, given the conditions in the definition above, one can construct a quasi-local channel circuit that reverses the action of $\calC$, which further satisfies local reversibility. This is because a local operation $\calE$'s influence on a finite-Markov-length state can be approximately reversed locally by another operation $\tilde\calE$ \cite{junge_universal_2018}, as depicted in Fig.~\ref{fig: circuit}(a).
Therefore, if we write the channel circuit $\calC$ (`forward circuit') as:
\begin{equation}
\begin{aligned}
    \calC=\calC_n \cdots \calC_1 \quad{\rm with}\quad
    \calC_t
    &={\prod}_x \calE_{t,x}
\end{aligned}
\end{equation}
then we can construct the following `reversal circuit':
\begin{equation}
\begin{aligned}
    \tilde\calC=\tilde\calC_1 \cdots \tilde\calC_n \quad
    {\rm with}
    \quad
    \tilde\calC_t
    &={\prod}_x \tilde\calE_{t,x}
\end{aligned}
\end{equation}
that (approximately) undoes $\calC$'s influence on $\rho_1$ in a gate-by-gate fashion, as long as the evolving state's Markov length remains finite throughout the evolution.

Due to the gate-by-gate or local reversibility, the forward and the reversal evolutions $\calC$ and $\tilde\calC$ are more special than a pair of evolutions satisfying only Def.~\ref{def: two_way}. In particular, they are \textit{localizable} and \textit{locality preserving}, as explained below.

{\bf Localizable}: Local reversibility implies that $\rho_1$ can be converted to and recovered from $\rho_2$ on different regions sequentially, as detailed below.  Suppose $\calC$ satisfies Def.~\ref{def: markov_path}.
For a partition of the lattice into $A$ and $B$ ($\Lambda = A \sqcup B$), one can partition $\calC$ and $\tilde\calC$ accordingly (as illustrated in Fig.~\ref{fig: circuit}(b, right)):
\begin{equation}
    \calC = \calC_A\calC_B ,
    \quad\quad
    \tilde\calC = \tilde\calC_B \tilde\calC_A ,
\end{equation}
such that $\calC_{A}$ is supported on $A$ together with the region near the boundary separating $A$ and $B$; and likewise for $\calC_B$.
Referring to figure~\ref{fig: circuit}(b), we can easily verify that:
\begin{equation}\label{eq: cancellation}
    \tilde\calC_A\calC_A\calC_B[\rho_1] = \calC_B[\rho_1],
    \quad\quad
    \tilde\calC_B\calC_B[\rho_1] = \rho_1
    ,
\end{equation}
which is also illustrated visually as follows:
\begin{equation*}
    \eqfig{3cm}{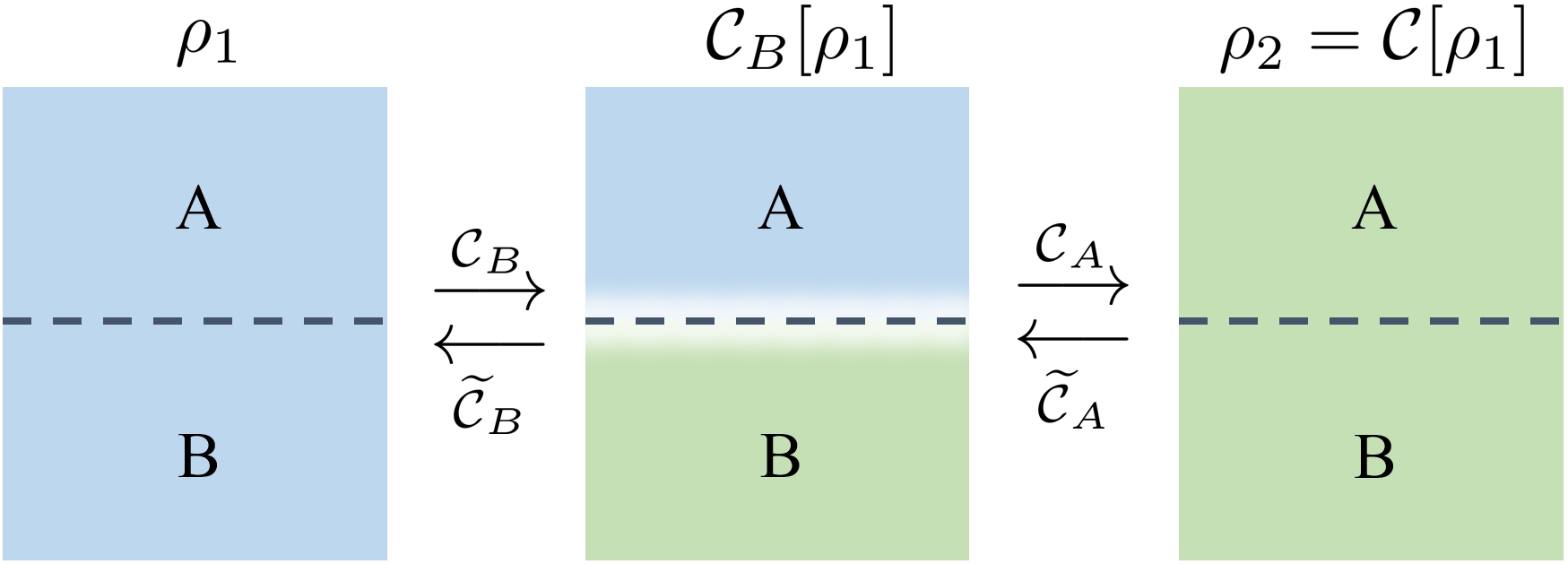}
\end{equation*}
In each panel, in the interior of the blue (green) region, the state looks $\rho_1$ ($\rho_2$). In the second panel, around the boundary region between $A$ and $B$ the state may be different from $\rho_1$ and $\rho_2$.

{\bf Locality preserving}: $\calC$ (and $\tilde\calC$) relate quantum operations on $\rho_1$ and those on $\rho_2$ in a locality-preserving fashion. Let $\calD$ be a quantum channel acting on $\rho_2$. It induces a channel $\tilde\calC\calD\calC$ acting on $\rho_1$, and its action on $\rho_1$ can be mimicked by another channel, $\calD^{\calC}$ (henceforth termed $\calD$ dressed by $\calC$), 
that is localized around the support of $\calD$:
\begin{equation}\label{eq: dressed_operation}
\begin{aligned}
    \tilde\calC\calD\calC[\rho_1] 
    &= \tilde\calC_{X}\tilde\calC_{\overline X}\calD\calC_{\overline X}\calC_{X}[\rho_1]\\
    &=\tilde\calC_{X}\calD\calC_{X}[\rho_1]
    \equiv
    \calD^{\calC}[\rho_1]
\end{aligned}
\end{equation}
where $X$ denotes the (slightly enlarged) spatial support of $\calD$. Equation \eqref{eq: dressed_operation} holds because the gates in $\tilde\calC_{\overline{X}}$ and $\calC_{\overline{X}}$ cancel each other due to Eq.\eqref{eq: cancellation}. The construction of $\calD^{\calC}$ is also illustrated graphically in Fig.\ref{fig: circuit}(c). We remark that in the special case that $\calC$ and $\tilde\calC$ are unitary circuits (\textit{i.e.} $\calC[\cdot]=U\cdot U^\dagger$), we have $\calD^\calC[\cdot] = U^\dagger\calD [U\cdot U^\dagger] U$, which manifestly has the same locality as $\calD$, just on a enlarged support. While for local channel evolution, we need the extra assumption of local reversibility to find a local $\calD^\calC$.

In the following sections, we will see that these properties make Def.~\ref{def: markov_path} a finer definition than Def.~\ref{def: two_way} and in particular, they distinguish $\rho_\myloop$ as a non-trivial phase.

\section{Topological degeneracy (TD) of mixed states}\label{sec: TD}
Topological degeneracy (TD) is a hallmark of zero-temperature quantum topological order. When defined on a topologically non-trivial lattice, a topologically ordered quantum Hamiltonian has a ground-state subspace $\calH_{G.S.}$ whose dimension depends on the lattice topology. Furthermore, different states within $\calH_{G.S.}$ are indistinguishable using any local observables.

Mixed states of interest are not necessarily associated with any Hamiltonian's low-energy states. Thus, when generalizing TD to mixed states, it is most natural to use local indistinguishability as the starting point (similar notions for mixed-states were proposed in \cite{li2024replica, ellison2024towards}). 

\begin{figure*}
    {(a)}\\
    \includegraphics[width=0.9\linewidth]{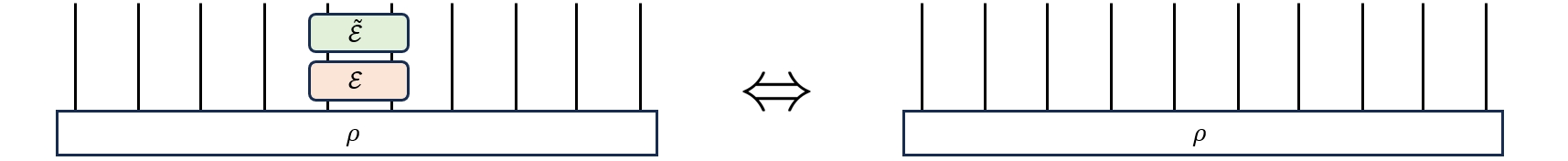}\\
    \vspace{.6cm}
    {(b)}\\
    \includegraphics[width=0.9\linewidth]{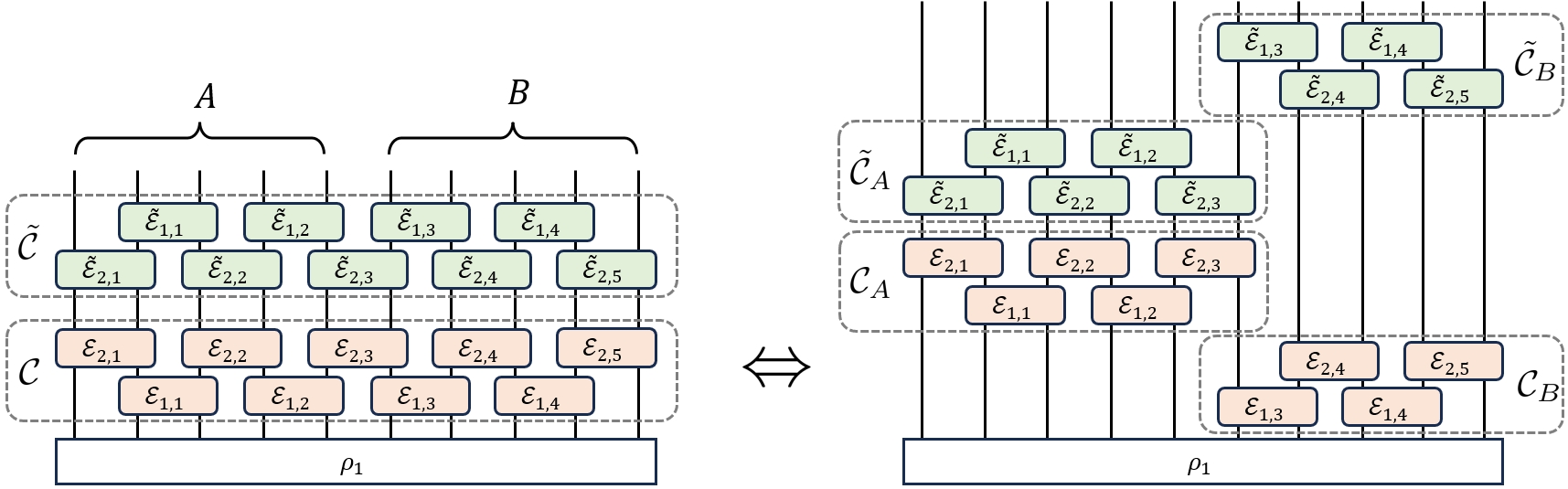}\\
    \vspace{.6cm}
    {(c)}\\
    \includegraphics[width=0.9\linewidth]{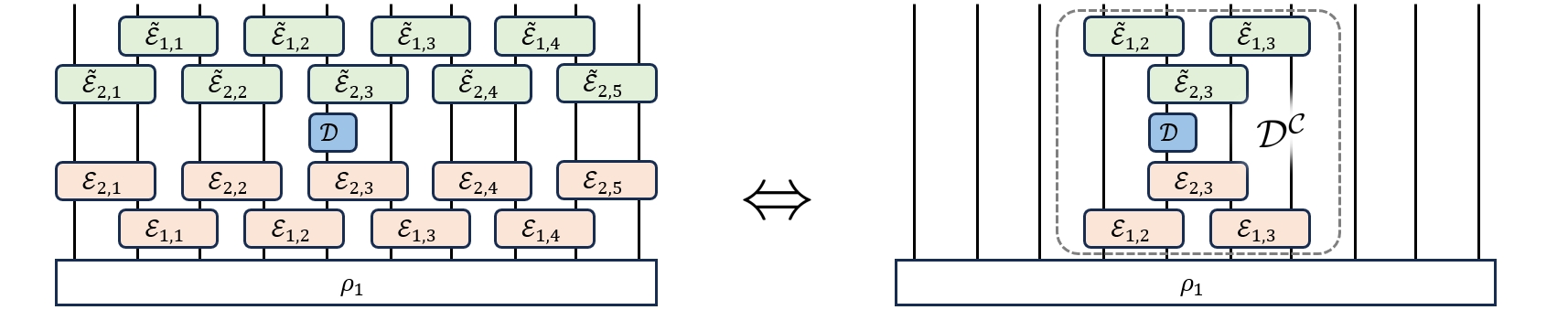}
    \caption{ \label{fig: circuit}
    (a) Illustration of local reversibility of a single gate. We say a channel gate $\calE$ (red box) is locally reversible with respect to the input state $\rho$ if there exists another channel $\tilde\calE$ (green box) localized around $\calE$ such that $\tilde\calE\calE\rho=\rho$. Each leg represents both bra and ket Hilbert spaces. If $\rho$ has a finite Markov length $\xi(\rho)$, then any local channel gate $\calE$ is locally reversible with respect to it.
    (We remark that $\tilde\calE$'s support is usually larger than $\calE$'s (see~\cite{sang2024stability} for more details). Here we draw them to be of the same size for the sake of illustration.)
    (b) \textit{l.h.s.} depicts a two-layer circuit $\calC$ and its local reversal circuit $\tilde\calC$. The latter exists if $\rho_1$'s Markov length remains finite at any time under the evolution of $\calC$. Each gate $\calE_{t,x}$ (in $\calE$) and its corresponding $\tilde\calE_{t,x}$ satisfy the local reversibility condition illustrated in (a) upon the state at the corresponding layer. Thanks to this gate-by-gate reversal property, for any spatial bi-partition into $A$ and $B$, one can partition the circuit accordingly (\textit{r.h.s.}) such that Eq.\eqref{eq: cancellation} is satisfied.
    (c) Due to the local reversibility, action of a channel $\tilde\calC\calD\calC$ upon $\rho_1$ (\textit{l.h.s}) can be mimicked by another channel $\calD^{\calC}$ whose support is slightly enlarged.
 }
\end{figure*}

\subsection{Mixed-state topological degeneracy (TD)}
\begin{definition}[Locally indistinguishable set]\label{def: LI_set}
For a mixed state $\rho$ defined on a region $\Lambda$, its locally indistinguishable set $\calQ(\rho;\xi_0)$ is the set of states that are identical to $\rho$ on any topologically trivial sub-region and have a Markov length no greater than $\xi_0$:
\begin{equation}
\begin{aligned}
\calQ(\rho;\xi_0) \equiv 
\{
&\sigma:\ \xi(\sigma)\leq\xi_0,\ \ \sigma_A\eqL\rho_A\\
&\text{for all simply connected } A\subseteq\Lambda
\}
\end{aligned}
\end{equation}
\end{definition}
The definition resembles the definition of the information convex set in entanglement bootstrap (see \textit{e.g.} Def C.1 in~\cite{SHI2020168164}). 
A few remarks about the definition:
\begin{enumerate}[leftmargin=.5cm]
\item The finite-Markov-length constraint is necessary for the definition to be meaningful.
As we will illustrate with the $\rho_\myloop$ example, the constraint removes `unphysical' locally indistinguishable states. These unphysical degeneracies are unstable and do not carry over to other states in the same phase.
We expect that $\calQ(\rho;\xi_0)$'s dependence on $\xi_0=O(1)$ is not important in the thermodynamic limit when $\xi_0$ is chosen to be sufficiently large (in particular, $\xi_0$ needs to be larger than $\xi(\rho)$). We henceforth assume that $\xi_0$ is always chosen to be large enough and drop the $\xi_0$ within $\calQ(\rho;\xi_0)$. 
\item $\calQ(\rho)$ defines a convex set, \textit{i.e.}, if $\sigma_1, \sigma_2 \in \calQ(\rho)$, then $\forall p\in[0,1], \sigma_3\equiv p\sigma_1+(1-p)\sigma_2\in\calQ(\rho)$. To see this, it suffices to check that $\sigma_3$ is locally indistinguishable from $\rho$ and has a finite Markov length. The former follows from: $\sigma_{3,A}=p\sigma_{1,A}+(1-p)\sigma_{2,A}=\rho_A $ for any $A$. To verify the latter, we notice that: $I_{\sigma_3}(A:C|B) = S_{\sigma_3}(A|B)-S_{\sigma_3}(A|BC)$, where $S(X|Y)\equiv S(XY)-S(Y)$ is the quantum conditional entropy. Since $A\cup B$ is a simply connected region in the definition of Markov length, $S_{\sigma_3}(A|B)\eqL S_{\sigma_{1,2}}(A|B)$. We cannot treat $S(A|BC)$ similarly because $C$ is not simply connected in general (\textit{e.g.}, when $ABC$ forms a torus). We instead apply the concavity of conditional entropy: $S_{\sigma_3}(A|BC)\geq p S_{\sigma_1}(A|BC) +(1-p) S_{\sigma_2}(A|BC)$, which gives us $I_{\sigma_3}(A:C|B)\leq p I_{\sigma_1}(A:C|B) +(1-p) I_{\sigma_2}(A:C|B)$, for sufficiently large $L$. Thus, $\xi(\sigma_3)$ is finite and bounded by $\max\{\xi(\sigma_1), \xi(\sigma_2)\}$.
\item Reflexivity: If $\rho\in\calQ(\sigma)$, then $\sigma\in\calQ(\rho)$.
\item A state $\sigma$ belongs to $\calQ(\rho)$ if and only if for any simply connected region $A$ away from $\Lambda$'s boundary, there exists a channel $\calD$ acting on $\overline A$ only satisfying $\calD[\rho]=\sigma$. To prove the `if' claim, we let $B$ be a width $r$ buffer region surrounding $A$, and $C$ be the rest of the system. Then by definition of $\calQ(\rho)$ we have $\rho_{AB} = \sigma_{AB}$. Furthermore, since $\sigma$ have finite Markov length $\xi$, when $r\gg \log L$ one can find a channel $\calR_{B\rightarrow BC}$ such that $\calR_{B\rightarrow BC}(\sigma_{AB})=\sigma$~\cite{junge_universal_2018}. Therefore letting $\calD=\calR_{B\rightarrow BC}\circ\tr_{C}$ fullfills our goal. To prove the `only if' claim, we first notice that since $\calD$ can be pushed out of any local region $A$, $\sigma$ and $\rho$ must be indistinguishable on $A$. Further, since $I_{\calD[\rho]}(A:C|B)\leq I_{\rho}(A:C|B)$ for any $\calD$ acting on $D$ only, we have $\xi(\sigma)\leq\xi(\rho)$. 

\end{enumerate}

We say a mixed state $\rho$ defined on $\Lambda$ \footnote{We assume $\rho$ can be consistently defined on any closed manifold $\Lambda$ of the same dimension. The $\rho_\myloop$ satisfies this condition.} has \textbf{topological degeneracy} (TD) if the shape of its $\calQ(\rho)$ as a convex set depends non-trivially on the topology of $\Lambda$.
Mixed-state TD generalizes the well-studied ground state degeneracy of quantum topological order. For instance, consider the 2D toric code Hamiltonian:
\begin{equation}
    H_{\rm t.c.} = -{\sum}_{v}A_v - {\sum}_p  B_p
\end{equation}
where $A_v$ and $B_p$ were defined in Eq.\eqref{eq: gibbs_rep_of_rholoop} and Eq.\eqref{eq: lindbladian_pflip}.
When defined on a 2D closed manifold with genus $g$, the Hamiltonian has a degenerate ground state subspace $V$ with a Hilbert space dimension of $2g$. As we will show, the locally indistinguishable set $\calQ(\ketbra{\psi}{\psi})$ of any ground state $\ket{\psi}\in V$ is given by:
\begin{equation}
    \calQ(\ketbra{\psi}{\psi}) = \conv\{\ketbra{\phi}{\phi}:\ \ket{\phi}\in V\},
\end{equation}
which is the set of all density matrices supported on the subspace $V$. To prove this claim, we first notice that all states within $A$ are identical on any simply connected region $A$, \textit{i.e.}, $\tr_{\overline A}(\sketbra{\psi})$ is independent of $\ket{\psi}$. Second, to verify that the states have a finite Markov length, it suffices to check it for pure states $\sketbra{\psi}$ within $\calQ$. As we remarked earlier, for pure states, we have $I(A:C|B) = I(A:C)=0$. Therefore, all states within $\calQ(\sketbra{\psi})$ have zero CMI and thus zero Markov length.

Non-topologically-ordered degenerate ground states, \textit{e.g.} those of spontaneous symmetry breaking of a 0-form symmetry, do not have non-trivial $\calQ$ because different states in the subspace can be distinguished locally through an order parameter.

\subsection{TD of the classical loop state}\label{sec: TD_of_rho_loop}
We now study the topological degeneracy of the classical loop state $\rho_\myloop$ defined on lattice $\Lambda$ being a torus.

To understand the structure of $\calQ(\rho_\myloop)$, we can draw insights from its similarity to the toric code state $\ket{\rm t.c.}$. While $\rho_\myloop$ is an incoherent mixture of loop configurations, $\ket{\rm t.c.}$ is their coherent superposition. In the toric code, distinct ground states are characterized by the parity of non-contractible loops around the two holes of the torus. This suggests that a similar mechanism might generate a non-trivial $\calQ(\rho_\myloop)$, which indeed proves to be the case.

We define four loop states with definite winding number parity:
\begin{equation}
    \rho_{\myloop}^{s_x s_y} 
    \propto
    \sum_{\bs\in{\rm loops}}\delta(f_x(\bs)=s_x, f_y(\bs)=s_y)\sketbra{\bs}
\end{equation}
where $f_{x/y}(\bs)$ measures the parity (modulo 2) of non-contractible loops around the x/y direction, and $s_x, s_y\in\{0, 1\}$. The original loop state $\rho_\myloop$ is just $\rho_{\myloop}^{00}$.

To demonstrate that each $\rho_{\myloop}^{s_x s_y}$ belongs to $\calQ(\rho_\myloop)$, we employ the stabilizer formalism. Given a set $\calG=\{g_1,...,g_m\}$ of commuting and (multiplicatively) independent $n$-qubit Pauli operators, its corresponding stabilizer state $\rho$ is defined as the maximally mixed state in the simultaneous $+1$ eigensubspace of each $g_i$:
\begin{equation}
    \rho \propto\prod_{g_i\in\calG}\left(\frac{1+g_i}{2}\right)
\end{equation}

Here, each $g_i\in\calS$ is called a stabilizer generator (or simply stabilizer) of $\rho$. The generating set of a stabilizer state is not unique. For example, $\calG'=\{g_1 g_2, g_2, g_3, ..., g_m\}$ defines the same stabilizer state as $\calG$. Two generating sets $\calG$ and $\calG'$ define the same state if and only if they generate the same abelian group of Pauli operators, known as the stabilizer group of $\rho_\calG$.

Each $\rho_\myloop^{s_x, s_y}$ is a stabilizer state. They share a common set of vertex stabilizers (see Eq.\eqref{eq: gibbs_rep_of_rholoop} for the definition of $A_v$) 
     $\{A_v\}_{v} $
which enforce the loop constraint at each vertex\footnote{One vertex must be excluded from this set due to the dependency relation $\prod_v A_v = 1$}. Additionally, each $\rho_\myloop^{s_x s_y}$ has two stabilizers that fix the parities of non-contractible loops in the $x$ and $y$ directions:
\begin{equation}
    \calG(\rho_{\myloop}^{s_x s_y}) =\{A_v\}_{v}  \cup \{(-1)^{s_x} Z_{\overline{\omega}_{x}}, (-1)^{s_y} Z_{\overline{\omega}_{y}}\}
\end{equation}
where $\overline{\omega}_{x/y}$ represents any non-contractible loop along the $x$/$y$ direction on the dual lattice (see Fig.\ref{fig: stabilizer} for illustration). The operator $Z_\gamma = \prod_{i\in\gamma} Z_i$ measures the 'flux' of loops across a closed loop $\gamma$. Different choices of $\gamma$ with the same topology yield equivalent states, as $\gamma$'s shape can be modified by multiplying with vertex stabilizers $A_v$.

It is straightforward to verify that all $\rho_\myloop^{s_x s_y}$ states are indistinguishable on any simply connected region $A$. We now examine the state's Markov length. For a stabilizer state $\rho$, the conditional mutual information $I(A:C|B)$ is given by:
\begin{equation}
\begin{aligned}
    I(A:C|B) = \min_{\calG(\rho)}|\{g\in \calG(\rho),
    \text{support of $g$}\\ \text{ intersects $A$ and $C$}\}|
\end{aligned}
\end{equation}
where the minimum is taken over all equivalent choices of $\calG(\rho)$~\cite{sanghybrid}. For a partition $\Lambda=A\cup B\cup C$ with the topology shown in Fig.\ref{fig: stabilizer}, local $A_v$ terms do not contribute to $I(A:B|C)$ when $B$'s width exceeds a few lattice spacings. For the non-local stabilizers $Z_{\overline{\omega}_{x,y}}$, we can always choose paths $\gamma$ that are supported only in $C$, as only their topology matters. Therefore, $\xi(\rho_{\myloop}^{s_x s_y})=0$ and
\begin{equation}
    \mathsf{conv}\{\rho_{\myloop}^{00}, \rho_{\myloop}^{01}, \rho_{\myloop}^{10}, \rho_{\myloop}^{11}\} \subseteq \calQ(\rho_{\myloop})
\end{equation}

\begin{figure}[h!]
    \centering
    \includegraphics[width=0.5\linewidth]{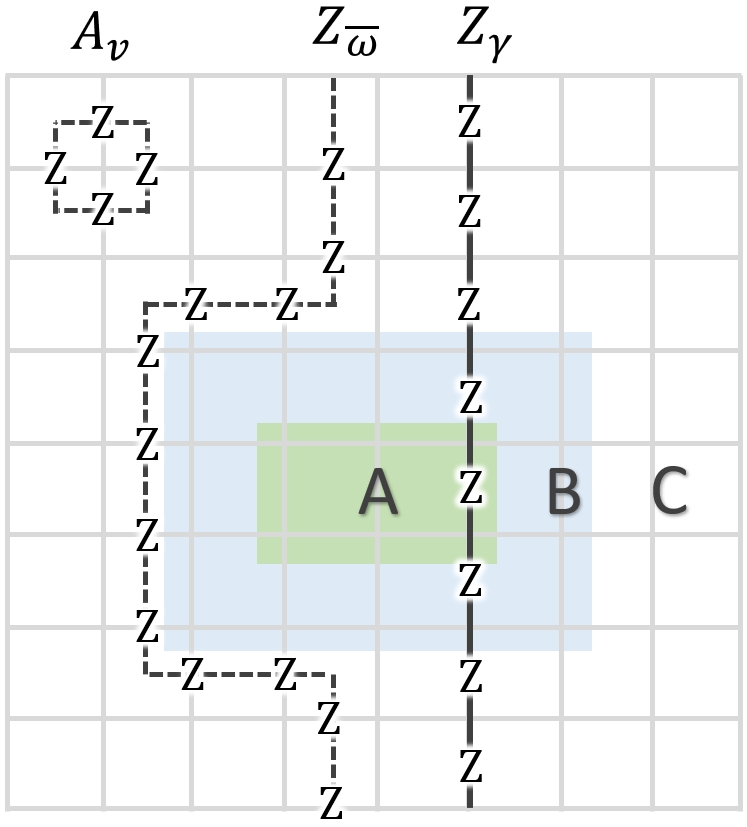}
    \caption{Construction of a stabilizer state $\sigma$  belonging to $\calQ(\rho_{\myloop})$.}
    \label{fig: stabilizer}
\end{figure}

We now argue that this inclusion is an equality. 
Consider any stabilizer state in $\calQ(\rho_{\myloop})$, which must be locally indistinguishable from $\rho_\myloop$ and have finite Markov length. All $A_v$ terms must be retained as stabilizers; otherwise, the loop constraint would be violated at some vertices, contradicting local indistinguishability 
Moreover, no new stabilizer whose support lies within a simply connected region $X$ can be introduced, as this would alter $\rho_X$, again violating local indistinguishability.

Therefore, to satisfy local indistinguishability, we must add stabilizers with non-simply-connected support. While there are many possible choices, such as Pauli operators of the form $\prod_{i\in \kappa}Z_\sigma$ where $\kappa$ is non-simply-connected, most of these choices increase the state's CMI and violate the finite-Markov-length condition.

Consider, for example, adding $Z_{\gamma}$ where $\gamma$ is a non-contractible loop on the \textit{primal} lattice (not the dual lattice) as shown in Fig.\ref{fig: stabilizer}. Unlike $Z_{\overline\omega}$, the shape of $Z_\gamma$ cannot be modified by multiplying with $A_v$ terms. Consequently, for any region $A$ that $\gamma$ crosses (as illustrated in Fig.\ref{fig: stabilizer}), $Z_\gamma$ increases the state's $I(A:C|B)$ by 1 for any $B$, resulting in an infinite Markov length $\xi$. This argument extends to all non-simply connected stabilizers wrapping around the torus, except for the `topological' operators $Z_{\overline{\omega}}$ that can be deformed by multiplying with $A_v$ terms.

This analysis generalizes to $\rho_\myloop$ on closed surfaces with higher genus $k$. The corresponding locally indistinguishable set $\calQ(\rho_{\myloop})$ forms a simplex with $2^{2k}$ extremal points, with its structure determined solely by the surface topology. Each extremal point corresponds to a loop ensemble with fixed winding number parity around each 'hole' of the surface (a genus-$k$ surface has ${2k}$ inequivalent holes).


\subsection{TD as a phase invariant}

We now examine whether topological degeneracy (TD) is a universal property shared by all states within a mixed-state phase.

We notice that if one adapts the two-way definition of mixed-state phase (Def.\ref{def: two_way}), the answer is clearly negative: on a torus $\rho_\myloop$ and $\ket{\bf 0}$ belong to the same phase using this definition (Sec.\ref{sec: defs}), while the latter does not display TD (\textit{i.e.} the product state's locally indistinguishable set contains only itself).

We will now demonstrate that, under the Markov length definition of phase equivalence (Def.\ref{def: markov_path}) TD is a phase invariant. More precisely, let $\calC[\rho]$ be any state in the same phase as $\rho$ with $\calC$ being the circuit that keeps the Markov length finite throughout, then $\calC$ and $\tilde\calC$ establish a one-to-one correspondence between states in $\calQ(\rho)$ and those in $\calQ(\calC[\rho])$:
\begin{equation}\label{eq:TD_invariance_comm_diagram}
\calQ(\rho)\overunderset{\calC}{\tilde\calC}{\rightleftharpoons}\calQ(\calC[\rho])
\end{equation}

Consider a local channel circuit $\calC=\calC_n...\calC_1$ acting on $\rho_0 \equiv\rho$ that keeps its Markov length finite. By definition, all $\rho_t \equiv \calC_t[\rho_{t-1}]$ states are in the same phase, and (recall discussion in Sec.\ref{sec: markov_length}) there is another channel circuit $\tilde\calC=\tilde\calC_1...\tilde\calC_n$ that reverses $\calC$'s action on $\rho$.

We first prove that for any $\sigma_0\in\calQ(\rho)$,  evolving $\sigma_0$ using $\calC$ also keeps its Markov length finite throughout the evolution. It suffices to check this for a single time step $\calC_1$. The evolved state can be written as:
\begin{equation}
        \sigma_1 
        \equiv
        \calC_{1}[\sigma_0]
        =\calC_{1}\calD[\rho_0]
        = \calC_{1}\calD\tilde\calC_1\calC_1[\rho_0]
        = \calD^{\calC_1}[\rho_1]
\end{equation}
where $\calD$ is a quantum channel satisfying $\sigma_0=\calD[\rho_0]$ (see comment 4 underneath Def.\ref{def: LI_set}).  For any partition $\Lambda=A\cup B\cup C$ following the geometry in Fig.\ref{fig: main}(b), we can choose $\calD$'s (and consequently $\calD^{\calC_1}$'s) support to lie entirely within $C$ (see the comment 4 beneath Def.\ref{def: LI_set}). 
By the data processing inequality \footnote{We comment that the inequality can be enhanced to an equality, since there also exists another channel $\calD'$ satisfying $\rho_0=\calD'[\sigma_0]$. For our purpose only the inequality is needed.}:
\begin{equation}\label{eq: cmi_mono}
    I_{\calD^{\calC_1}[\rho_1]}(A:C|B)\leq I_{\rho_1}(A:C|B)
\end{equation}
holds for any $A\cup B\cup C$ partition, implying $\xi(\sigma_1)\leq\xi(\rho_1)$. Iteratively applying this argument to the full evolution $\calC_{n}...\calC_1$, we obtain:
\begin{equation}
    \xi(\sigma_t)\leq\xi(\rho_t)\ \text{with}\ \sigma_t=\calC_t...\calC_1[\sigma_0]\ \ \forall\sigma_0\in \calQ(\rho_0)
\end{equation}

This condition, combined with the local indistinguishability of $\rho_t$ and $\sigma_t$, ensures that $\tilde\calC$ reverses $\calC$'s action on $\sigma_0$:
\begin{equation}
    \tilde\calC\calC[\sigma_0] = \sigma_0
\end{equation}

Since this holds for any $\sigma_0\in\calQ[\rho_0]$, $\calC$ (and $\tilde\calC$ ) establish an isometric bijection between $\calQ(\rho_0)$ and $\calQ(\calC[\rho_0])$. Specifically, for any pair of states $\sigma, \sigma'\in\calQ(\rho)$, and for $D$ being any distance measure of quantum states satisfying monototicity under quantum operation (\textit{e.g.} trace distance, relative entropy, fidelity distance), we have:
\begin{equation}
\begin{aligned}
&D\left(\calC[\sigma], \calC[\sigma']\right)
\leq 
D\left(\sigma, \sigma'\right)\\
= 
&D\left(\tilde\calC\calC[\sigma], \tilde\calC\calC[\sigma']\right)
\leq
D\left(\calC[\sigma], \calC[\sigma']\right)\\
\Rightarrow\quad\quad&
D\left(\calC[\sigma], \calC[\sigma']\right) = D\left(\sigma, \sigma'\right)
\end{aligned}
\end{equation}
where both inequalities follow from the monotonicity.

We now apply the result above to the classical loop state defined on a torus.
Since $\calQ[\rho_\myloop]$ in this case is a tetrahedron, $\calQ(\calC[\rho_\myloop])$ must also be a tetrahedron due to linearity of quantum channels. This geometric preservation excludes the possibility that $\rho_\myloop$ and the product state $\ket{0...0}$ belong to the same phase on a torus, as the latter has a trivial $\calQ$ containing only itself. 

We remark that this phase-inequivalence proof assumes that the total lattice has a non-zero genus (otherwise both $\calQ(\rho_{\myloop})$ and $\calQ(\sketbra{0..00})$ are trivial, thus the no-go condition does not hold). In Sec.\ref{sec: anomaly}, we present another proof that applies to $\rho_\myloop$ defined on zero-genus manifold, \textit{e.g.} sphere and disk.

\subsection{Preservation of memory encoded in the TD and coherent information}

As an application of the results above, we now consider the scenario in which the TD of $\calQ(\rho)$ serves as memory for some encoded information, which is then subjected to noise. We formalize this scenario as follows: a TD memory of $\rho$ is a triple $(\mathcal{M}, \calE, \calD)$, where $\calM$ is a convex set of logical (mixed) states that we want to encode; an encoding channel $\calE : \calM \to \calQ(\rho)$; and a decoding channel $\calD : \calQ(\rho) \to \calM$ satisfying $\calD \circ \calE = \identity_{\calM}$. $\calM$ can be both classical (a simplex) or quantum (e.g. Bloch ball for a qubit).

Given the definition above, the preservation of the TD memory of $\rho$ is a direct consequence of the phase invariance of the locally indistinguishable set $\calQ(\rho)$, argued in Sec. \ref{sec: TD_of_rho_loop} for the classical loop state. Indeed, if $\calC(\rho)$ is another state in the same phase as $\rho$, then it has a TD memory $(\calM, \calC \circ \calE, \calD \circ \tilde\calC)$ due to the following commuting diagram (See Eq. \ref{eq:TD_invariance_comm_diagram}):
\begin{equation}
\calM 
\overunderset{\calE}{\calD}{\rightleftharpoons}
\calQ(\rho)
\overunderset{\calC}{\tilde\calC}{\rightleftharpoons}
\calQ(\calC[\rho])
\end{equation}

For a quantum memory $\calM$, the amount of quantum information preserved after a channel $\calC$ acts on a state $\sigma \in Q(\rho)$ can be captured by the coherent information $I(\sigma, \calC)$ \cite{nielsen_quantum_2010}. It is defined as $I(\sigma, \calC) = S(Q') - S(RQ')$, where $\sigma$ is purified to an extended system $RQ$, which is then acted upon by $\identity_R \otimes \calC_Q$ to become $RQ'$. Hence, it measures the amount of correlations still present between $R$ and $Q$. By the data processing inequality, the coherent information achieves its maximum value $S(\sigma)$ if, and only if, the state can be perfectly recovered, which is precisely the case if $\calC$ is a locally reversible channel. For the pure toric code case, this conclusion agrees with previous explicit calculations of the coherent information \cite{fan2023diagnostics, lee_exact_2025}.

\section{Spontaneous breaking of weak 1-form symmetry}\label{sec: anomaly}

For quantum topological orders, topological degeneracy arises from the existence of anyons and their non-trivial braiding relations. In modern terminology, quantum topological order (in 2D) results from spontaneous breaking of 1-form symmetries—symmetry operators that take the form of a loop.

In this section, we demonstrate that the non-trivialness of $\rho_\myloop$ and its corresponding mixed-state phase emerge from a similar mechanism: the breaking of a weak 1-form symmetry, and its corresponding fractionalized excitations. Remarkably, the symmetry action takes the form of a quantum channel when away from the fixed-point state. 

The classical loop state exhibits the following symmetries (recall that $X_\gamma = {\prod}_{i\in\gamma} X_i$ and $Z_{\overline{\omega}}={\prod}_{i\in\overline{\omega}} Z_i$):
\begin{equation}
\begin{aligned}
    \rho_\myloop &= \calU_\gamma[\rho_\myloop] \equiv X_\gamma \rho_\myloop X_\gamma \\
    \rho_\myloop &= Z_{\overline{\omega}} \rho_\myloop
\end{aligned}
\end{equation}
where $\gamma$ and $\overline{\omega}$ represent any closed loop of edges on the lattice and dual-lattice, respectively. Note that $X_\gamma$ must act on both sides to preserve the density matrix, while $Z_{\overline{\omega}}$ needs to act only on one side. In the literature, these are termed a \textit{weak 1-form symmetry} and a \textit{strong 1-form symmetry}, respectively.


The strong and weak $1$-form symmetries of $\rho_\myloop$ exhibit a 't Hooft anomaly in the following sense. When we truncate a loop $\gamma$ into an open string $\gamma_{ab}$ between vertices $a$ and $b$, $\calU_{\gamma_{ab}}$ creates two non-local point-like excitations at vertices $a$ and $b$ by flipping the signs of $A_a$ and $A_b$. The resulting state becomes a mixture of closed loops with singular endpoints at $a$ and $b$.

These excitations are non-local in the sense that no operation localized near $a$ and $b$ can remove them to restore $\rho_\myloop$ from $\calU_{\gamma_{ab}}[\rho_\myloop]$. To demonstrate this, consider any dual loop $\overline{\omega}$ that encircles only vertex $a$ while remaining well-separated from both $a$ and $b$. We find:
\begin{equation}
    Z_{\overline{\omega}}\cdot\calU_{\gamma_{ab}}[\rho_{\myloop}]
    =-\calU_{\gamma_{ab}}[\rho_{\myloop}]
    \label{eq:EManomaly}
\end{equation}
This follows from the fact that $\gamma_{ab}$ and $\overline\omega$ intersect once, implying $X_{\gamma_{ab}}Z_{\overline{\omega}}=-Z_{\overline{\omega}}X_{\gamma_{ab}}$. The relation Eq.~\eqref{eq:EManomaly} is the manifestation of the 't Hooft anomaly and is also known in the literature as  braiding statistics.

The anomaly structure precludes the possibility of locally removing the excitations around $a$ and $b$, as any such local operation cannot alter $Z_{\overline{\omega}}$'s expectation value from $-1$ back to $\tr(Z_{\overline{\omega}} \rho_\myloop)=+1$, given that $\overline\omega$ is distant from both vertices. Furthermore, only the topology of $\gamma_{ab}$ is relevant: if $\gamma'_{ab}$ is another open string connecting $a$ and $b$, then $\calU_{\gamma'_{ab}}[\rho_{\myloop}]=\calU_{\gamma_{ab}}[\rho_{\myloop}]$.

It is also useful to view $X_\gamma$ as a weak symmetry operation on the state and $Z_{\overline{\omega}}$ as an observable, where $\tr Z_{\overline{\omega}}\rho_\myloop=1$. The anomaly structure then implies that the weak $1$-form symmetry $X_\gamma$ is spontaneously broken, with $Z_{\overline{\omega}}$ being the loop order parameter. Alternatively, one can view $Z_{\overline{\omega}}$ as a strong symmetry operation and $X_\gamma$ as a ``fidelity observable'' in the sense that $F(\rho_\myloop,\calU_\gamma[\rho_\myloop])=1$. The anomaly structure then implies a strong-to-weak spontaneous symmetry breaking (SWSSB) of the $1$-form symmetry $Z_{\overline{\omega}}$, with $X_\gamma$ being the loop order parameter~\cite{Zhang2025SWSSB}.

In this work, we will use several phrases interchangeably: the mutual 't Hooft anomaly between the strong and weak $1$-form symmetries, the mutual braiding statistics, SSB of the weak $1$-form symmetry, and SWSSB of the strong $1$-form symmetry.\footnote{Strictly speaking, the order parameter of a sponaneously broken symmetry may not always give another symmetry. For example, for the weak symmetry SSB, it may be the case that one can find a loop order parameter $\tilde{Z}$ such that $\tr\tilde{Z}\rho=c$ where $0<c<1$, and there is no choice of order parameter that can make $c=1$. We will not consider such cases in this work. In fact our current understanding supports the conjecture that in a gapped state, with exponentially decaying mutual information and conditional mutual information, a spontaneous breaking of (strong or weak) $1$-form symmetry will always lead to a dual $1$-form symmetry. This conjecture is the natural generalization of the ``modularity'' condition of pure state topological orders.}



\begin{figure}
    \centering
    \includegraphics[width=\linewidth]{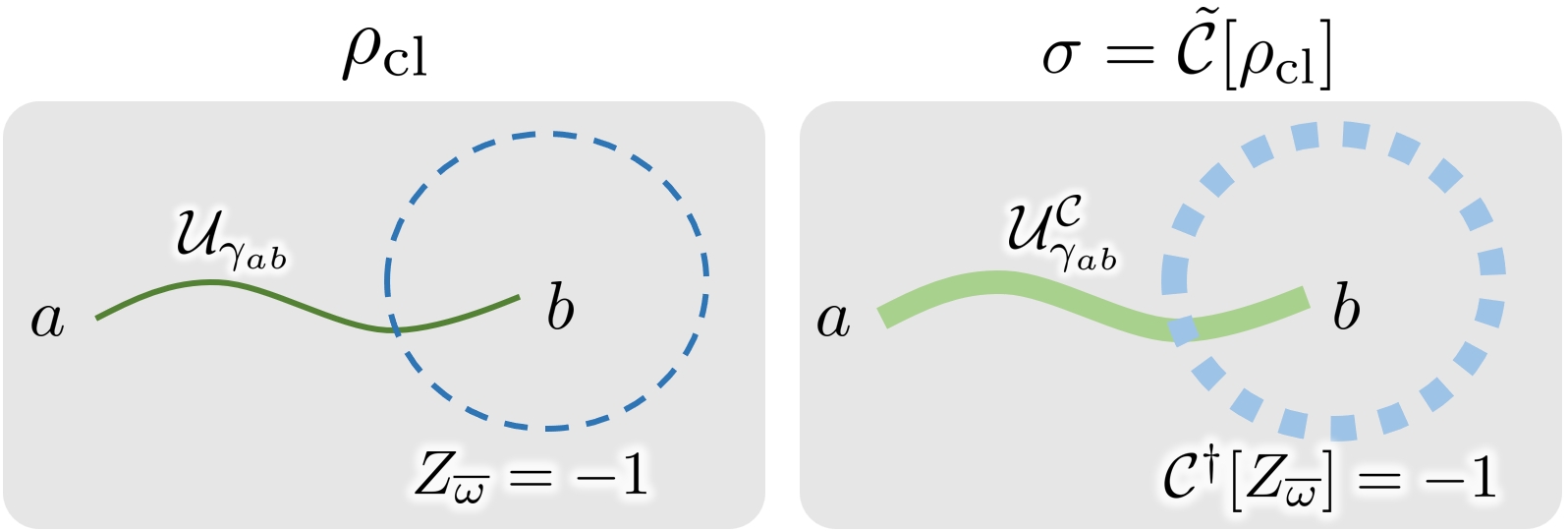}
    \caption{
    Illustration of 't Hooft anomaly / SSB of 1-form symmetries for the loop state $\rho_{\myloop}$ (left), and for a generic state $\sigma$ in same phase as $\rho_\myloop$. In both cases, truncating the weak symmetry action from a closed string to a open one (solid green line) creates two point-like excitations on the two ends $a$ and $b$. Such excitations can be detected non-locally by an operator (`strong symmetry', dashed blue line) wrapping around it. The weak symmetry $\calU^{\calC}$ acting on $\sigma$ takes the form of a channel symmetry generically.
      }
    \label{fig: detection}
\end{figure}


We will show that these symmetry-breaking properties extend to all states in the same phase as $\rho_\myloop$ (under the Markov length definition). Each state exhibits both `dressed' strong and weak 1-form symmetries. Away from the fixed point, these symmetries may be deformed: the weak symmetry becomes a ``channel symmetry''-- a quantum channel supported on a smeared loop, while the strong symmetry becomes a Hermitian operator on a smeared dual loop. When truncated, the weak symmetry action continues to create point-like excitations at its endpoints, detectable by the strong symmetry.

One can obtain explicit expressions of these dressed symmetries. For a state $\sigma$ in the same phase as $\rho_{\myloop}$, there exists a circuit $\tilde\calC$ and its local reversal $\calC$ such that $\sigma = \tilde\calC[\rho_\myloop]$ and $\calC[\sigma]=\rho_\myloop$, as discussed in Sec.\ref{sec: markov_length}. $\sigma$ possesses the following dressed weak and strong 1-form symmetries:
\begin{equation}\label{eq:dressed_symms}
\begin{aligned}
    \sigma &= \calU_\gamma^{\calC}[\sigma]\\
    \sigma &= \calC^\dagger[Z_{\overline{\omega}}]\, \sigma
\end{aligned}
\end{equation}
where $\calU_\gamma^{\calC}$ was defined in Eq.\eqref{eq: dressed_operation} (and Fig.\ref{fig: circuit}(c)), and $\calC^\dagger$ is the dual map of $\calC$ defined by the relation $\tr(A\,\calC[B])=\tr(\calC^\dagger[A]B)$.  The first equality follows directly from Eq.\eqref{eq: dressed_operation}:
$\calU_\gamma^{\calC}[\sigma] = \tilde\calC\,\calU_\gamma\,\calC[\sigma] = \sigma.$
The second equality derives from:
$1=\tr(Z_{\overline{\omega}}\,\calC[\sigma]) = \tr(\calC^\dagger[Z_{\overline\omega}]\sigma)$
which implies $\calC^\dagger[Z_{\overline\omega}]\sigma=\sigma$ by Lemma.\ref{lemma: lemma} in the App.\ref{ap: lemma}.\footnote{A related dressing of strong symmetries was investigated in Refs. \cite{lessa2025higher} and \cite{ellison2024towards}, with former calling it ``symmetry pullback''. There, however, the dressed symmetry is defined on an enlarged Hilbert space, preserving its unitarity. After tracing out the ancilla, it coincides with our definition via the dual channel.}


Note that the dressed weak symmetry in general has the action of a channel, and the fact that $\calC$ is locally reversible is crucial in preserving the locality of this channel symmetry.  This is essential in ensuring that the dressed weak 1-form symmetry $\calU_\gamma^{\calC}$ on $\sigma$ also exhibits SSB:
\begin{equation}\label{eq: dressed_sym_detection}
\calC^\dagger[Z_{\overline{\omega}}]\cdot\calU_{\gamma_{ab}}^{\calC}[\sigma]=-\calU_{\gamma_{ab}}^{\calC}[\sigma]
\end{equation}
where the configurations of $\gamma_{ab}$ and $\overline\omega$ are illustrated in Fig.\ref{fig: detection}. We postpone the proof to Sec. \ref{sec: proof}. Importantly, equation \eqref{eq: dressed_sym_detection} above it means that anyons are well-defined throughout the phase of $\rho_\myloop$; they are non-local excitations created at the endpoints of the string-like channel $\calU_{\gamma_{ab}}^{\calC}$ that can be detected by loop operators $\calC^\dagger[Z_{\overline\omega}]$ encircling them. 

One way to quantifiably measure the existence of excitations created by $\calU_{\gamma_{ab}}^{\calC}$ is by computing the relative entropy between $\sigma$ and $\calU_{\gamma_{ab}}^{\calC}[\sigma] = \calC[\calU_{\gamma_{ab}}(\rho_\myloop)]$ \cite{fan2023diagnostics}. The reason is that the relative entropy $S(\sigma || \calU_{\gamma_{ab}}^{\calC}[\sigma])$ is high when the two states are approximately distinguishable to each other, which is to be expected if $\calU_{\gamma_{ab}}^{\calC}$ creates detectable excitations on $\sigma$. Indeed, since $\sigma$ and $\calU_{\gamma_{ab}}^{\calC}[\sigma]$ are supported in distinct eigenspaces of $\calC^\dagger[Z_{\overline{\omega}}]$ by Eq. \eqref{eq:dressed_symms}, their relative entropy is $+\infty$. Had we taken care of the approximation errors in the local recovery equation $\tilde\calC \calC(\rho_{\myloop}) \approx \rho_{\myloop}$ and in the exponential tails of the reverse channel $\tilde\calC$, we expect the relative entropy to scale with the distance between the excitation points $a$ and $b$, instead of being infinite. This behavior was found in \cite{fan2023diagnostics} for the case of the pure toric code state under Pauli $Z$ and $X$ dephasing below noise threshold.




\subsection{TD on an annulus and topological entropy}
\label{sec: TD_TE}
So far, we have shown that $\rho_\myloop$ has a non-trivial topological degeneracy $Q(\rho_\myloop)$, and that this is a \textit{universal} feature of its phase. Namely, that any other state $\sigma$ in the same phase as $\rho_\myloop$ also has a non-trivial topological degeneracy $Q(\sigma)$, and that it is in one-to-one correspondence with $Q(\rho_\myloop)$. Here, we present two additional universal features of $\rho_\myloop$ that distinguish it from any states in the trivial phase, namely its non-trivial degeneracy on a annulus shaped subsystem, and a lower-bound on topological entropy. Both properties depend only on $\rho_{\myloop}$'s reduced state on a large subregion. We remark that discussions below apply to any state with spontaneously broken weak 1-form symmetry.

We again assume $\sigma=\tilde\calC[\rho_{\myloop}]$ is a state in the same phase as $\rho_\myloop$, both defined on some 2D lattice $\Lambda$. Let $\Gamma\subset\Lambda$ be an annulus-shaped region with both its width and radius much larger than $\calC$'s range. We investigate the locally indistinguishable set $\calQ(\sigma_\Gamma)$, where $\sigma_\Gamma=\tr_{\overline\Gamma}\sigma$.

Let $\gamma_{ab}$ be a string connecting point $a$ in the interior of $\Gamma$ to point $b$ outside $\Gamma$, and let $\overline{\omega}$ be a dual loop wrapping around $\Gamma$, as illustrated below:
\begin{equation*}
    \eqfig{2.4cm}{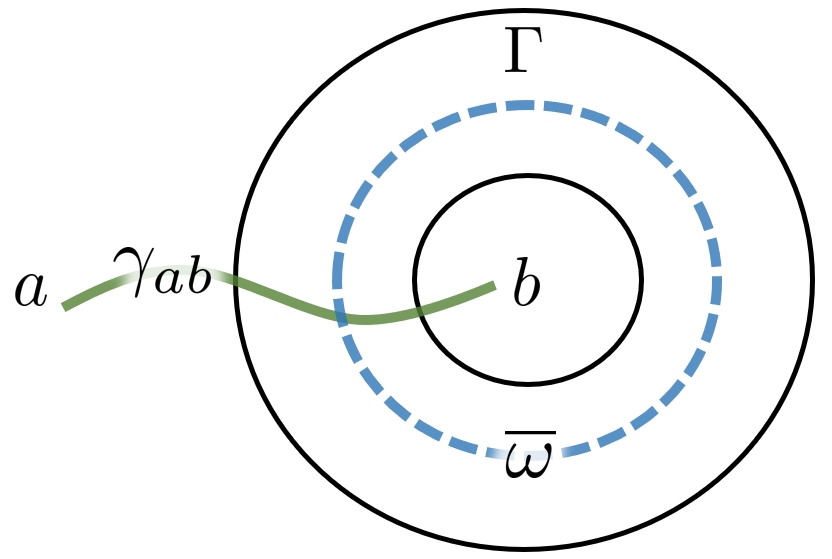}
\end{equation*}

Consider the following state defined on $\Gamma$:
\begin{equation}
    \sigma'_\Gamma \equiv \tr_{\overline\Gamma}\left(\calU^\calC_{\gamma_{ab}}[\sigma]\right).
\end{equation}

The state $\sigma'_\Gamma$ is locally indistinguishable from $\sigma_\Gamma$ on any simply connected region $A\subset\Gamma$, as $\gamma_{ab}$ can always be deformed away from $A$ without altering the resulting state $\calU^\calC_{\gamma_{ab}}[\sigma]$. Furthermore, for a closed non-contractable loop $\overline\omega$ within $\Gamma$ that is away from the boundary:
\begin{equation}
\begin{aligned}
    &\calC^\dagger[Z_{\overline{\omega}}]\cdot\calU_{\gamma_{ab}}^{\calC}[\sigma]=-\calU_{\gamma_{ab}}^{\calC}[\sigma]\\
    \Rightarrow
   \quad&\calC^\dagger[Z_{\overline{\omega}}]\sigma'_{\Gamma}=-\sigma'_{\Gamma}
\end{aligned}
\end{equation}

Since $\sigma_\Gamma$ satisfies $\calC^\dagger[Z_{\overline{\omega}}]\sigma_{\Gamma}=\sigma_{\Gamma}$, the support of $\sigma'_\Gamma$ in Hilbert space must be orthogonal to that of $\sigma_\Gamma$. (Recall that $\calC^\dagger[Z_{\overline \omega}]$ is a hermitian operator as the dual map is hermiticity preserving.)
Therefore, $\calQ(\sigma_\Gamma)$ is a non-trivial set:
\begin{equation}
    \mathsf{conv}\{\sigma'_\Gamma, \sigma_\Gamma\}\subseteq\calQ(\sigma_\Gamma).
\end{equation}

This provides another proof that $\rho_{\myloop}$ belongs to a non-trivial phase: if we take $\sigma=\sketbra{00...0}$, then $\sigma_{\Gamma}$ is simply a product state with a trivial $\calQ$. Remarkably, this proof imposes no condition on the global topology of the lattice.

The degeneracy also leads to a lower-bound on the topological entropy -- another hallmark of topological phases. It has long been noticed that $\rho_\myloop$ has a non-zero topological entropy~\cite{castelnovo2007topological}. Here we show that it is indeed a universal feature shared by all states in $\rho_\myloop$'s mixed-state phase. The proof below is inspired by a similar proof for pure states in \cite{levin_physical_2024}.

Using the Levin-Wen tripartition~\cite{PhysRevLett.96.110405} of an annulus into $A, B$ and $C$, as shown in the figure below, the topological entropy is then given by the conditional mutual information $I_\sigma(A:C|B)$. In addition, we let $A'$ be a region shrunk radially compared to $A$, and $R$ be a region that covers the part of $\gamma_{ab}$ within $A$:
\begin{equation*}
    \eqfig{3.cm}{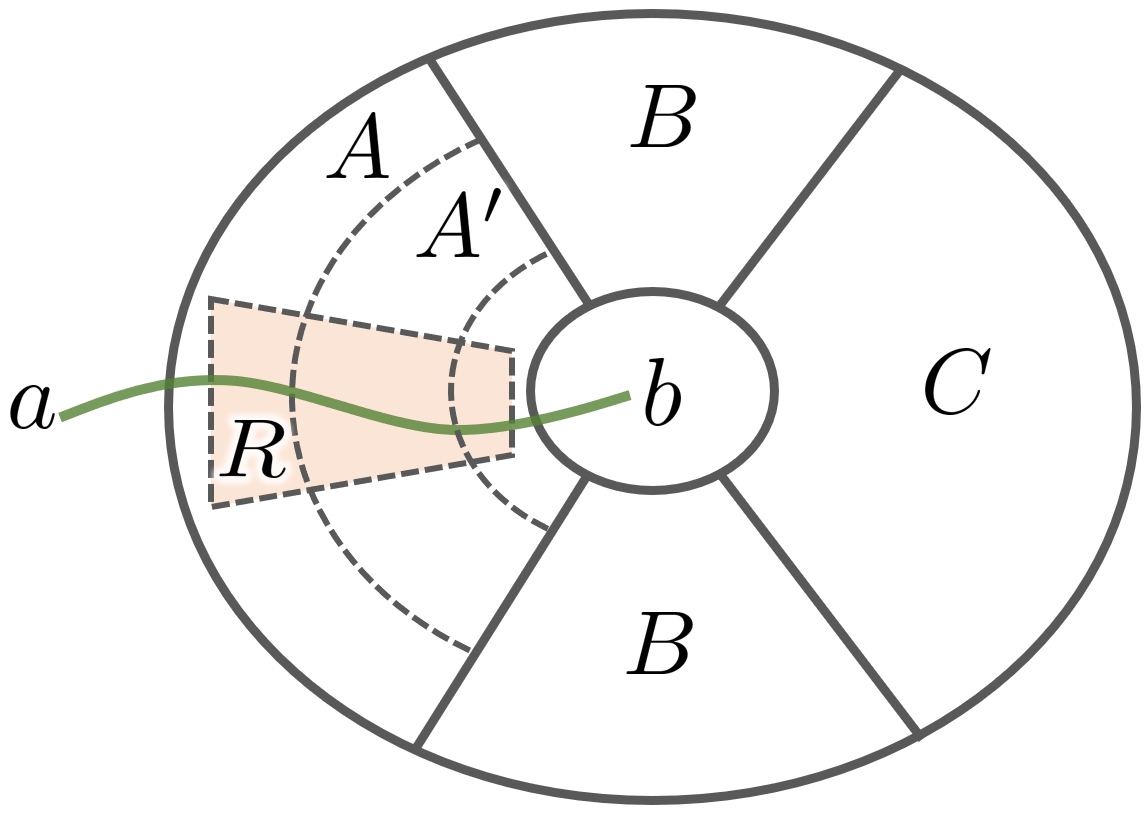}
\end{equation*}
Due to the monotonicity of CMI under channels acting on $A$ or $C$ only,
\begin{equation}
    I_{\sigma}(A:C|B)\geq I_{\calC_R[\sigma]}(A:C|B)\geq I_{\calC_R[\sigma]}(A':C|B)
\end{equation}
where $\calC_R$ is part of $\calC$ as defined through decomposition $\calC=\calC_{\overline R}\calC_R$ (see Eq.\eqref{eq: cancellation}). Recall that $\calC_R[\sigma]$ is indistinguishable from $\rho_\myloop$ in the interior of $R$.

We define a new state $\sigma''=\frac{1}{2}\sigma' + \frac{1}{2}\sigma\in\calQ(\sigma_\Gamma)$, which is the equal mixture of $\sigma$ and $\sigma'$. 
We have:
\begin{equation}
    \begin{aligned}
        &I_\sigma(A:C|B)\\
        \geq& I_{\calC_R[\sigma]}(A':C|B)\\
        =& I_{\calC_R[\sigma'']}(A':C|B) + S_{\calC_R[\sigma'']}(A'BC) - S_{\calC_R[\sigma]}(A'BC)\\
        \geq& S_{\calC_R[\sigma'']}(A'BC) - S_{\calC_R[\sigma]}(A'BC)
    \end{aligned}
\end{equation}
the first inequality is due to the data-processing inequality for CMI, while the last inequality is because CMI must be non-negative. Since $\calC_R[\sigma'']_{A'BC}=\frac{1}{2}\calC_R[\sigma]_{A'BC} + \frac{1}{2}\calC_R[\sigma']_{A'BC}$ is an equal mixture of two states being orthogonal to each other, we get:
\begin{equation}\label{eq: S_diff}
    \begin{aligned}
        &S_{\calC_R[\sigma'']}(A'BC) - S_{\calC_R[\sigma]}(A'BC)\\
        =
        &\frac{1}{2}\left(S_{\calC_R[\sigma']}(A'BC) - S_{\calC_R[\sigma]}(A'BC)\right)_{=0} + \log2\\
        = 
        &\log2.
    \end{aligned}
\end{equation}
And therefore
\begin{equation}
    I(A:C|B)\geq \log2.
\end{equation}
In Eq.\eqref{eq: S_diff}, the term in the bracket is zero because the two states $\calC_R[\sigma]_{A'BC}$ and $\calC_R[\sigma']_{A'BC}$ are related by an onsite unitary action, namely a product of $X$ along the segment of $\gamma_{ab}$ within $A'$. 

We remark that following a similar argument as in~\cite{kim2013long}, a non-zero topological entropy leads to a lower bound on the level of topological degeneracy on torus. In this sense the lower-bound above is related to our conclusion in Sec.\ref{sec: TD_of_rho_loop}.

\subsection{\texorpdfstring{Proof of Eq.~\eqref{eq: dressed_sym_detection}}{Proof of preservation of anomaly under weak symmetry dressing}}\label{sec: proof}
We now go back to prove SSB of dressed 1-form symmetry in $\sigma=\tilde\calC[\rho_\myloop]$, namely Eq.\eqref{eq: dressed_sym_detection}.

It suffices to show that $\tr(l.h.s.\text{ of Eq.}\eqref{eq: dressed_sym_detection})=-1$ (according to App.\ref{ap: lemma}). Using Eq.\eqref{eq: dressed_operation} and the definition of the dual map, we have $\tr(l.h.s)=\tr(Z_{\overline\omega}\cdot\calC\,\tilde\calC\,\calU_{\gamma_{ab}}[\rho])$. Let $A$ be a region that intersects with $\overline\omega$ and is distant from $\gamma_{ab}$, and let $B$ be the complement of $A$. This gives rise to the circuit partition $\calC=\calC_B\calC_{A}$. Let $\gamma'_{ab}$ be a curve connecting points $a$ and $b$ that crosses region $A$. The geometry of $A$, $\gamma_{ab}$, $\gamma'_{ab}$, and $\overline\omega$ is illustrated below:
\begin{equation*}
    \eqfig{2.5cm}{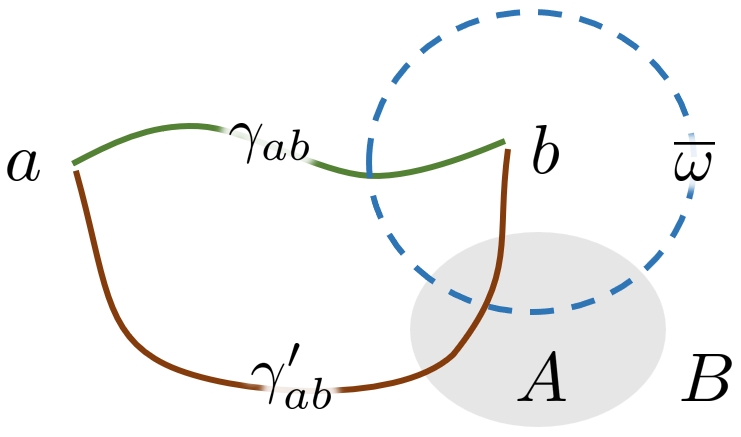}
\end{equation*}

Since $\gamma_{ab}$ is well separated from $A$, Eq.\eqref{eq: cancellation} implies $\calC\,\tilde\calC\,\calU_{\gamma_{ab}}[\rho]=\calC_B\tilde\calC_B\calU_{\gamma_{ab}}[\rho]$. Additionally, as $\gamma_{ab}\cup\gamma'_{ab}$ forms a closed loop, we have $\calU_{\gamma_{ab}}[\rho]=\calU_{\gamma'_{ab}}[\rho]$. Therefore:
\begin{equation}
    \begin{aligned}
        \tr(\text{\textit{l.h.s.} of Eq.\eqref{eq: dressed_sym_detection}})
        &=\tr\left(Z_{\overline\omega}\cdot\calC_{B}\tilde\calC_B\calU_{\gamma_{ab}}[\rho]\right)\\
        &=\tr\left(\tilde\calC^\dagger_B\calC^\dagger_B[Z_{\overline\omega}]\cdot\calU_{\gamma'_{ab}}[\rho]\right).
    \end{aligned}
\end{equation}

A key observation is that $\tilde\calC^\dagger_B\calC^\dagger_B[Z_{\overline\omega}]$ takes the form $Z_{{\overline\omega}_A}\otimes\tilde\calC^\dagger_B\calC^\dagger_B[Z_{\overline{\omega}_B}]$, where ${\overline\omega}_{A(B)}$ denotes the portion of $\overline\omega$ within region $A$ ($B$). Moreover, $\tilde\calC^\dagger_B\calC^\dagger_B[Z_{\overline{\omega}_B}]$ has support around $\overline{\omega}_B$. This leads to the following commutation relation:
\begin{equation}
\begin{aligned}
    \tilde\calC^\dagger_B\calC^\dagger_B[Z_{\overline\omega}]\cdot\calU_{\gamma'_{ab}}[...]
    &= (Z_{{\overline\omega}_A}\tilde\calC^\dagger_B\calC^\dagger_B[Z_{\overline{\omega}_B}])X_{\gamma'_{ab}}[...]X_{\gamma'_{ab}}\\
    &= -X_{\gamma'_{ab}}(Z_{{\overline\omega}_A}\tilde\calC^\dagger_B\calC^\dagger_B[Z_{\overline{\omega}_B}])[...]X_{\gamma'_{ab}}\\
    &= -\calU_{\gamma'_{ab}}\left[\tilde\calC^\dagger_B\calC^\dagger_B[Z_{\overline\omega}]...\right].
\end{aligned}
\end{equation}
The second equality holds because $\tilde\calC^\dagger_B\calC^\dagger_B[Z_{\overline{\omega}_B}]$ and $X_{\gamma'_{ab}}$ commute due to their non-overlapping supports. Therefore:
\begin{equation}
    \begin{aligned}
        \tr(\text{\textit{l.h.s.} of Eq.\eqref{eq: dressed_sym_detection}})
        &=-\tr\left(\calU_{\gamma'_{ab}}\left[\tilde\calC^\dagger_B\calC^\dagger_B[Z_{\overline\omega}]\cdot\rho\right]\right)\\
        &=-\tr\left(Z_{\overline\omega}\cdot\calC_B\tilde\calC_B[\rho]\right)\\
        &=-\tr\left(Z_{\overline\omega}\rho\right)= -1
    \end{aligned}
\end{equation}
This completes the proof of Eq.\eqref{eq: dressed_sym_detection}.


\section{Discussion and outlook}

We have proposed a new definition of mixed-state phases based on locally reversible channel circuits, and we have shown that local reversibility implies that operations on a state are transformed in a locality-preserving fashion.  As a result, all symmetries (strong and weak) and anomalies of the state are preserved under such channels.  The dressed weak symmetry naturally motivated the concept of a channel symmetry, a new direction to generalize symmetry in open quantum systems.  

We end with some discussions:
\begin{itemize}[leftmargin=.5cm]
\item The implications of channel symmetries, especially with respect to anomalies, deserve further investigation.  Such properties of mixed states are in sharp contrast to the case of pure state quantum phases, where dressed symmetries remain unitary. Moreover, the anomaly between dressed strong and weak symmetries is defined with respect to a particular state (as in Eq. \eqref{eq: dressed_sym_detection}), in contrast to the pure state case where anomaly is a property of symmetries only.  It is an important question whether, for two mixed states $\rho, \sigma$ in the same phase, unitary strong and weak symmetries of $\rho$ can induce dressed symmetries of $\sigma$ \textit{that are also unitary}, and can their anomaly hold at the operator level?  

\item As a primary application of our general results, we have demonstrated that the 2d classical loop ensemble is non-trivial under our new definition, despite being connected to a trivial product state via two-way local channels. It would be interesting to find other classes of states in which our refined definition leads to a different classification of phases.

\item In our definition, Markov length plays a central role in ensuring local reversibility. It has recently been proven that quantum Gibbs states of local Hamiltonians possess finite Markov length~\cite{chen2025quantum}. It is interesting to ask the converse question: if a state has finite Markov length, is it guaranteed to be a Gibbs state of a (quasi-)local Hamiltonian? The 1D special case of this question was addressed in \cite{kato2019quantum}, while in higher dimensions the answer appears to be unknown whenever the Markov length is not exactly zero.

\item Finally, we note that in~\cite{sang2023mixed}, a mixed-state real space renormalization group (RG) was proposed which ideally coarse-grains blocks of the system while preserving correlations between a block and its complement.  Importantly, such an ideal RG scheme constitutes a locally reversible channel which as shown in the current work preserves symmetries and anomalies of the state, as an RG flow should.
\end{itemize}

\noindent \textit{Note added:} While completing this manuscript, we became aware of a recent independent work \cite{yang2025topological} that discusses topological degeneracy and topological entropy in mixed states using the entanglement bootstrap approach. We would also like to draw readers' attention to a forthcoming work \cite{tba} discussing a different form of emergent channel symmetry acting on Gibbs states in open system dynamics.

\begin{acknowledgments}
We acknowledge helpful discussions with Tyler Ellison, Dominic Else, Yaodong Li, Isaac Kim, David P\'erez-Garc\'ia, Bowen Shi. 
S.S. and T.H. thank Isaac Kim and Zijian Song for collaborations on related topics. 
S.S. was supported by the SITP postdoctoral fellowship at Stanford University. L.A.L. acknowledges support from the Natural Sciences and Engineering Research Council of Canada (NSERC) under Discovery Grant No. RGPIN-2020-04688 and No. RGPIN-2018-04380.  This work was also supported by an Ontario Early Researcher Award.
R.M. is supported by the National Science Foundation under Grant No.~DMR-1848336.
T.G. is supported by the National Science Foundation under Grant No.~DMR-2521369. Research at Perimeter Institute is supported in part by the Government of Canada through the Department of Innovation, Science and Industry Canada and by the Province of Ontario through the Ministry of Colleges and Universities.
\end{acknowledgments}

\appendix

\section{Equivalent characterizations of strong symmetry}\label{ap: lemma}
We prove a lemma that was used multiple times in the main text. It concerns equivalent conditions for a state to posses a strong symmetry. A similar result appeared also in \cite{ellison2024towards, lessa2025higher}.
\begin{lemma}\label{lemma: lemma}
For a quantum state $\rho$, a operator $O$ satisfying $\norm{O}=1$, and a quantum channel $\calE$, the following three conditions are equivalent ($\lambda\in\mathbb{C}$ and $|\lambda|=1$):
\begin{enumerate}
    \item $O \cdot \calE[\rho] = \lambda\calE[\rho]$
    \item $\calE^\dagger[O]\cdot\rho = \lambda\rho$
    \item $\tr(\calE^\dagger[O]\cdot\rho) = \tr(O\cdot\calE[\rho]) = \lambda$
\end{enumerate}
\end{lemma}
\begin{proof}

    (1$\Leftrightarrow$2) Using the Stinespring dilation of the channel $\calE$:
    \begin{equation}
        \calE[\rho] = \tr_A[U(\rho\otimes\sketbra{\bf 0}_A)U^\dagger]
    \end{equation}
    Then we have:
    \begin{equation}
        \begin{aligned}
            &O \cdot \calE[\rho] = \lambda\calE[\rho]\\
            \Leftrightarrow\ \ 
            &O \cdot U(\rho\otimes\sketbra{\bf 0}_A)U^\dagger = \lambda U(\rho\otimes\sketbra{\bf 0}_A)U^\dagger\\
            \Leftrightarrow\ \ 
            &U^\dagger O U \cdot (\rho\otimes\sketbra{\bf 0}_A) = \lambda \rho\otimes\sketbra{\bf 0}_A\\
            \Leftrightarrow\ \ 
            &{\!}_A\!\braket{0 | U^\dagger O U | 0}_A \cdot \rho = \lambda \rho\\
            \Leftrightarrow\ \ 
            &\calE^\dagger[O]\cdot\rho = \lambda\rho
        \end{aligned}
    \end{equation}
    (1$\Leftrightarrow$3)
    Going from 1 to 3 is simply by taking partial trace. To go from 3 to 1, we consider an eigen-desomposition of the density matrix: $\calE[\rho] = \sum_{i}p_i\sketbra{\psi_i}$, with each $p_i\geq0$ and $\sum_i p_i = 1$. Then we have:
    \begin{equation}
    \begin{aligned}
        &1 = |\tr(O\calE[\rho])| = \left|{\sum}_i p_i \braket{\psi_i | O | \psi_i}\right|\\
        &\leq {\sum}_i p_i \left|\braket{\psi_i | O | \psi_i}\right| \leq {\sum}_i p_i= 1.
    \end{aligned}
    \end{equation}
    Therefore both inequalities are saturated. To saturate the first one, all $\braket{\psi_i | O | \psi_i}$ must have the same phase. To saturate the second, each $\left|\braket{\psi_i | O | \psi_i}\right|=1$. Combining with the condition that $\tr(O\calE[\rho])=\lambda$, we know all $\ket{\psi_i}$'s are eigenstates of $O$ with the eigenvalue $\lambda$. We thus obtain:
    \begin{equation}
        O\cdot\calE[\rho] = \lambda \calE[\rho]
    \end{equation}
\end{proof}

\section{Example of two-way connected states in different phases} \label{appendix:two-way_different_phases}

The Markov length definition of phase equivalence (Definition \ref{def: markov_path}) requires not only that two states $\rho$ and $\sigma$ in the same phase are two-way connected by local channels $\calE$ and $\calF$ -- $\calE(\rho) = \sigma$ and $\calF(\sigma) = \rho$ -- but also that $\calE$ is locally reversible by $\calF$. Here, we exemplify the importance of this extra condition with the case of the classical loop state $\rho_{\myloop}$ and the trivial state $\trivialstate = \identity / \dim(\calH)$. We show they are two-way connected via continuous one-parameter families of channels $\calE_p$ and $\calF_q$ (generated by Lindbladians), even though they are not in the same phase, since $\rho_{\myloop}$ has classical memory but the trivial state doesn't. Hence, the Markov length must diverge along both paths, which we also show explicitely via mappings to statistical models. 

To go from $\rho_{\myloop}$ to $\rho_{\mytrivial}$, we apply $X$-dephasing channel $\calE_p = \calD_{X}(p) = \prod_e [(1-p) (\cdot) + p X_e (\cdot)X_e]$ for $p \in [0,1/2]$:
\begin{equation}
    \rho^{\myloop\to\mytrivial}(p) = \prod_e [(1-p) \rho_{\myloop} + p X_e \rho_{\myloop} X_e]
\end{equation}
At maximal dephasing strength $p = 1/2$, $\rho_{\mytrivial} \propto \identity$ is reached because the strong Pauli-$Z$ 1-form symmetries of $\rho_{\myloop}$, which completely characterize it, become weak. 

The Markov CMI of $\rho^{\myloop\to\mytrivial}(p)$ can be shown to be equal to the one of the $X$-dephased pure toric code \cite{dennis2002topological, fan2023diagnostics, bao2023mixed, lee_quantum_2023}, which is known have Markov length diverging at $p_c \approx 0.11$ by mapping to the random bond Ising model \cite{sang2024stability}. Thus, we conclude that there is a transition from the classical memory phase to a trivial phase at $p=p_c$.

\subsection{From the trivial state to the classical loop state}
\label{appendix:trivial_to_classical}


Alternatively, to continuously transition from $\trivialstate=\identity/\dim(\calH)$ to $\rho_\myloop$, we first replace $\trivialstate$ with the state $\ket{\bf 0} \equiv \ket{00\cdots0}$, also in the trivial phase, via the replacement channel $\mathcal{R}_{\ket{0}} \equiv \ketbra{\bf 0}{\bf 0} \tr[\cdot]$, and then dephase the plaquette terms $B_f = \prod_{e \in f} X_e$ via $\mathcal{D}_{B}(q) = \prod_{f} [(1-q)(\cdot) + q B_f (\cdot) B_f]$, resulting in the channel $\mathcal{F}_q = \calD_{\text{plaq}}(q) \circ \calR_{\ket{\bf 0}}$:
\begin{align}
    \rho^{\mytrivial\to\myloop}(q) & = \mathcal{F}_q(\rho_{\mytrivial}) \\ 
    & = \prod_{f} [(1-q) \ketbra{\bf 0}{\bf 0} + q B_f \ketbra{\bf 0}{\bf 0} B_f] \label{eq:tr_to_cl_def}
\end{align}
for $q \in [0,1/2]$. For $q=1/2$, the classical loop state with no non-contractible loops $\rho_{\myloop}=\rho_{\myloop}^{00}$ is reached. To prepare anoter state in the classical memory $Q(\rho_\myloop)$ (e.g. $\rho^{10}_{\myloop} = X_{\bar{\omega}_x} \rho_{\myloop}^{00} X_{\bar{\omega}_x}$)one just needs to act with the appropriate Wilson loop operator ($X_{\bar{\omega}_x}$) in $\ket{\bf0}$.

By calculating the Markov CMI and the topological entropy (See Fig.\ref{fig:plaq}), we find numerical evidence that $\rho^{\mytrivial\to\myloop}(q)$ is in the same (trivial) phase as $\rho^{\mytrivial\to\myloop}(0) = \rho_{\mytrivial}$ for $q < 1/2$. The intuition for this is that $\rho^{\mytrivial\to\myloop}(q)$ is a mixture over all loops with weights exponentially decaying with the \emph{area} of the loop, and not its boundary, for $q<1/2$, which heavily suppresses long-range correlations necessary for a topological classical memory.

\begin{figure*}[tb]
    \centering
        \includegraphics[width=0.45\textwidth]{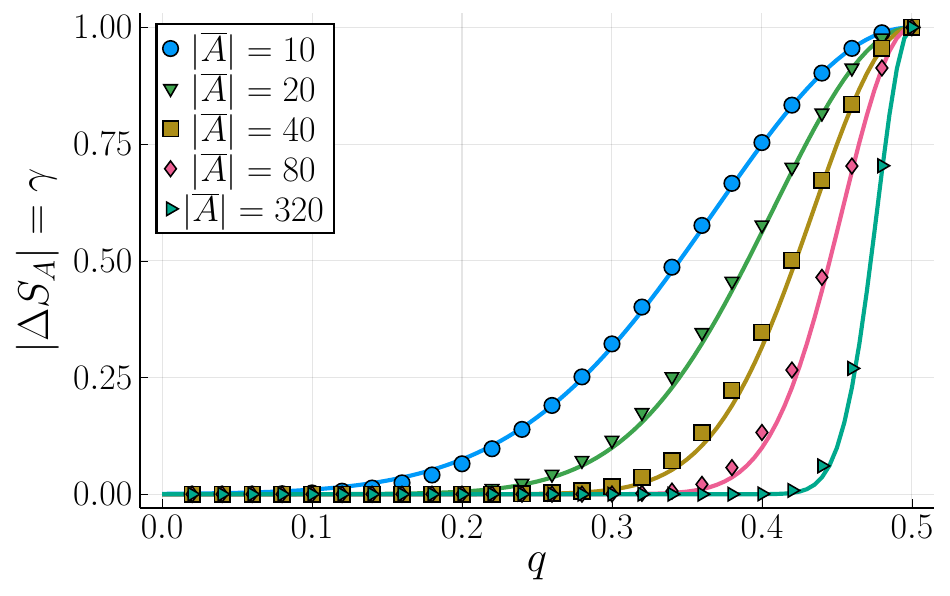}
        \label{fig:plaq_Delta_S}
        \hspace{1cm}
        \includegraphics[width=0.45\textwidth]{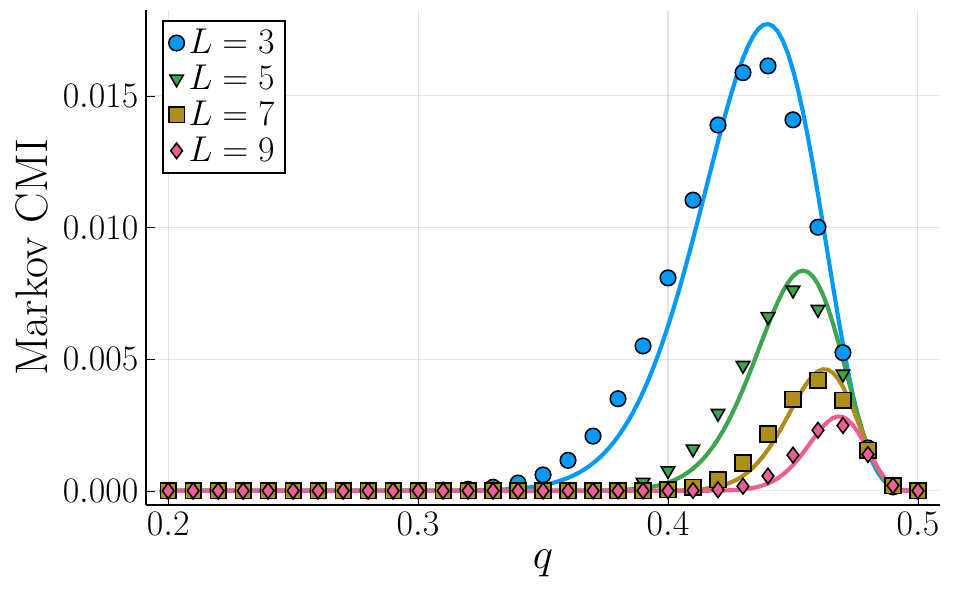}
        \label{fig:plaq_Markov_CMI}
    \caption{Entropic measures of $\rho^{\mytrivial\to\myloop}(q)$ that signal the phase transition at $q_c = 1/2$: (Left) the non-extensive correction $\Delta S_A(q)$ to the entanglement entropy, whose magnitude matches the topological entropy $\gamma$ in the thermodynamic limit; and (Right) the CMI $I(A:C|B)$ in the Markov tripartition for different widths $L$ of $B$, where $A$ and $C$ have fixed length. The solid lines were evaluated using the analytical approximation for $\Delta S_A(q) = S_A(q) - S_A^{\text{pm.}}(q)$ given by Eq. \eqref{eq:DeltaS_approx}, whereas the data points were numerically computed from the exact expression, Eq. \eqref{eq:DeltaS_exact}.
    }
    \label{fig:plaq}
\end{figure*}

We can calculate these quantities due to the following expression for $\rho^{\mytrivial\to\myloop}(q)$:
\begin{align}
    \rho^{\mytrivial\to\myloop}(q) & = \sum_{\sigma \in C_2(M)} (1-q)^{N_f - |\sigma|} q^{|\sigma|} \ketbra{\partial \sigma}{\partial \sigma} \\
    & \propto \sum_{\sigma \in C_2(M)} \lambda^{|\sigma|} \ketbra{\partial \sigma}{\partial \sigma},
\end{align}
where $\lambda = q / (1-q)$, and $C_2(M)$ is the space of all two-cochains of the base space $M$ with $\mathbb{Z}_2$ coefficients, i.e. $C_2(M) \simeq \{0,1\}^{N_f}$ with $N_f$ the number of faces. Given $\sigma \in C_2(M)$, we define the boundary state on the edges of $M$ as $\ket{\partial \sigma} = \bigotimes_{e} \prod_{f \ni e} \sigma_f$, and $|\sigma| = \sum_f \sigma_f$. The reduced density matrix to a subset of links $A$ is given by
\begin{align}
    \rho^{\mytrivial\to\myloop}_A(q) & \propto \sum_{\sigma \in C_2(\overline{A})} \lambda^{|\sigma|} \ketbra{r_A(\partial \sigma)}{r_A(\partial \sigma)},
\end{align}
where $r_A$ is the restriction map to $A$, and $C_2(\overline{A})$ is generated by all faces that have at least one edge in $A$. The sum above includes redundant states $\ket{\partial \sigma}$ iff $r_A(\partial \sigma) = r_A(\partial \sigma')$ for $\sigma \neq \sigma'$. This happens whenever $\sigma$ and $\sigma'$ differ by a cocycle $\eta \in Z_2(\overline{A})$. For a connected region $A$, for example, we have
\begin{align}\label{eq:reduced_loop_soup}
    \rho^{\mytrivial\to\myloop}_A(q) & \propto \sum_{\mathclap{\ell \in B_1(\overline{A})}} \left(\lambda^{|\sigma^*(\ell)|} + \lambda^{|\overline{A}| - |\sigma^*(\ell)|} \right) \ketbra{r_A(\ell)}{r_A(\ell)},
\end{align}
where $B_1(\overline{A}) = \{\partial \sigma \mid \sigma \in C_2(\overline{A}) \}$ is the set of all contractible loops $\ell$, and $\sigma^*(\ell)$ is an arbitrary 2-cochain whose boundary is $\partial \sigma^*(\ell) = \ell$. For disconnected $A$, $\rho_A^{\mytrivial\to\myloop}(\rho)$ is a tensor product over each connected component, as can be seen from the locality of $\calD_{\text{plaq}}$ in Eq. \eqref{eq:tr_to_cl_def}.

For $q < 1/2$ and large regions $|\overline{A}| \to \infty$, one of the $\lambda < 1$ factors in Eq. \eqref{eq:reduced_loop_soup} is exponentially suppressed compared to the other for almost all loops $\ell$. Thus, any quantity depending solely on the eigenvalues of $\rho_A^{\mytrivial\to\myloop}(q)$, such as its entropy $S_A(q)$, will follow the probability distribution of the 2d classical paramagnet on the plaquettes of $\overline{A}$. Since the paramagnet is a Gibbs ensemble, its entropy is extensive in system size, which implies zero (quantum) conditional mutual information $I(A:C|B)$ for any tripartition $A|B|C$ (See Fig. \ref{fig:plaq}). As such, the Markov length of $\rho_A^{\mytrivial\to\myloop}(q)$ does not diverge for $q < 1/2$, thus constituting a single phase.

For any finite size, however, the entropy $S_A(q)$ aquires non-extensive contributions even for $q < 1/2$, which in turn leads to a phase transition. We study its critical behavior in the next section.

\subsection{Critical behavior}
\label{appendix:critical_behavior}

With more work, we can also analyze the critical behavior near the critical point $q_c = 1/2$. Because of our interest in the CMI $I(A:C|B)$ for both the Markov and the Levin-Wen partitions, we will focus on the deviations of the entropy $S_A(q)$ from the extensive contribution, which gets canceled in the CMI. For that, we start from the expression for $S_A(q)$ coming from Eq. \eqref{eq:reduced_loop_soup} and divide it into two terms:
\begin{align}
    S_A(q) & = - \frac{1}{2\mathcal{Z}} \sum_{\ell \in B_1(\overline{A})} (\lambda^{|\sigma^*(\ell)|} + \lambda^{|\overline{A}| - |\sigma^*(\ell)|}) \times \\
    & \qquad \qquad \qquad \log(\lambda^{|\sigma^*(\ell)|} + \lambda^{|\overline{A}| - |\sigma^*(\ell)|}) \\
    & = - \frac{1}{\mathcal{Z}} \sum_{n=0}^{|\overline{A}|} \binom{|\overline{A}|}{n} \lambda^n \log(\lambda^n) - \\
    & \qquad \quad  \frac{1}{\mathcal{Z}} \sum_{n=0}^{|\overline{A}|} \binom{|\overline{A}|}{n} \lambda^n \log(1 + \lambda^{|\overline{A}|-2n}) \label{eq:DeltaS_exact}\\
    & = |\overline{A}| s_{A}^{\text{pm.}}(q) + \Delta S_A(q),
\end{align}
where $\mathcal{Z} = \frac{1}{2} \sum_{n=0}^{|\overline{A}|} \binom{|\overline{A}|}{n} (\lambda^n + \lambda^{|\overline{A}|-n}) = (1+\lambda)^{|\overline{A}|}$, and in the second line we assumed $A$ to be connected and transformed the sum over loops $\ell$ into a sum over the area variable $n \equiv |\sigma^*(\ell)|$. 

The first term above is extensive in $A$ and corresponds to the paramagnetic entropy. The second term, $\Delta S_A(q)$, is the deviation from the paramagnetic limit. Writing $\lambda = q/(1-q)$, we notice that $\Delta S_A(q)$ equals the expectation value of $ \log(1 + \lambda^{|\overline{A}|-2n})$ w.r.t. the binomial distribution $P(n) \propto  \binom{|\overline{A}|}{n} q^n (1-q)^{|\overline{A}|-n}$. We are interested in the behavior of  $\Delta S_A(q)$ for $q = q_c - \delta q$, for small $0 \leq \delta q \ll q_c$. In this limit, one finds that 
the binomial distribution is peaked at $n \approx |\overline{A}|/2 - |\overline{A}| \delta q$. Defining $n' = n - |\overline{A}|/2$, one therefore finds

	\begin{eqnarray} 
		 -\Delta S_A(q) & = & \langle \log(1 + e^{- 2n' \log(\lambda)})\rangle \\
		& \approx & \langle \log(1 + e^{8 n' dq})\rangle \\
		& \approx & \log(2) + 4 \langle n' \rangle dq + 8 \langle n'^2 \rangle (dq)^2 \\
        & \approx & \log(2) - 2 |\overline{A}| (dq)^2 \\
         & \approx &\log(2)e^{-|\overline{A}|/\xi(q)}, \label{eq:DeltaS_approx}
	\end{eqnarray}
where $\xi(q) = c (q_c-q)^{-\nu}$ diverges with critical exponent $\nu = 2$ as $q \to q_c$ ($c = \log(2)/2$). Above, all expectation values $\langle ...\rangle$ are w.r.t. to the aforementioned binomial distribution, and we have used  standard properties of this distribution to evaluate them for $\delta q \ll 1$. Contrary to the usual decay in length, $\xi(q)$ measures a correlation area. This is not an artifact of the approximation method, and comes from the fact that $S_A(q)$ depends on the region $A$ only via its area and number of connected components.

Since the volume and area law contributions for the entropy cancel out in the Markov tripartition, the Markov CMI is dictated by $\Delta S(q)$:
\begin{align}
    &I(A:C|B) \nonumber \\ & =  \Delta S_{BC}(q) + \Delta S_{AB}(q) - \Delta S_B(q) - \Delta S_{ABC}(q) \\
    & \approx \log(2)e^{-|\overline{B}|/\xi(q)} (1 - e^{-|\overline{A}|/\xi(q)})(1 - e^{-|\overline{C}|/\xi(q)})
\end{align}
From the expression above, we conclude that Markov CMI converges to zero for all $q \in [0,1/2]$, while the Markov length diverges as $q$ approaches $q_c =1/2$.

The topological entropy $\gamma$ can also be estimated. As a sublinear correction to the von Neumann entropy $S_A(q) = |\overline{A}| s_{A}^{\text{pm.}}(q) - \gamma(q)$, it is clear from Eq. \eqref{eq:DeltaS_approx} that $\gamma(q) \to 0$ for $q < 1/2$ and $\gamma(q_c) = \log(2)$, as expected for the classical loop state. The same result is found if $\gamma$ is calculated from the CMI in the Levin-Wen tripartition, up to a correction given by the Markov CMI for the same region sizes, which goes to zero in the thermodynamic limit.

To verify the predicted behavior near the critical point, we numerically calculate the non-extensive part of the von Neumann entropy $\Delta S_A(q) = S(\rho^{\mytrivial\to\myloop}_A(q)) - |\overline{A}| s_A^{\text{pm.}}(q)$, and the Markov CMI. For both, equation \eqref{eq:DeltaS_exact} was used, and the results are the solid lines of Fig. \ref{fig:plaq}. We were able to verify that the approximation given by Eq. \eqref{eq:DeltaS_approx} not only gives the correct critical behavior at $q \approx q_c$, but is also accurate even for $q / q_c \ll 1$ (and $|\overline{A}|$ large), far from the regime ($q \approx q_c$) it was initially derived from.



\subsection{Local non-reversibility and anomaly breaking}
\label{appendix:local_non-reversibility}

We finish the appendix by explaining in more details how $\rho_\mytrivial$ and $\rho_{\myloop}$ are two-way connected, even though they are in different phases. Apart from the divergence of the Markov length that was already discussed, another way to see that such two-way connection is not locally reversible is to compute what would the dressed strong and weak 1-form symmetries be for $\rho_{\mytrivial}$. If the channels taking one state to the other were locally reversible by one another, then the strong and weak symmetryes of $\rho_\myloop$ -- respectively, $Z_{\overline{\omega}}$ and $\calU_\gamma$ -- would become dressed symmetries of $\rho_{\mytrivial}$ maintaining the same anomalous braiding relations. However, according to the transformation rule of Eq. \eqref{eq:dressed_symms}, we would have
\begin{align}
    Z_{\overline{\omega}} & \to \calE_{p=1/2}^*(Z_{\overline{\omega}}) = \identity, \\
    \calU_\gamma & \to \calE_{p=1/2} \circ \calU_\gamma \circ \calF_{q=1/2} = \calR_{\identity}.
\end{align}
These two are valid strong and weak symmetries of the state $\rho_{\mytrivial} \propto \identity$, but they are trivially so. In particular, the anomalous braiding is lost. Moreover, the weak symmetry $\calU_\gamma$ is dressed to a channel with global support, instead of just around the loop $\gamma$. Both of these properties signal the failure of the two channels to be locally reversible by each other.


\bibliography{main.bib}
\end{document}